\documentclass[pra,reprint]{revtex4-1}
\usepackage[utf8]{inputenc}
\usepackage{amsmath}
\usepackage{amssymb}
\usepackage{amsthm}
\usepackage{hyperref}
\usepackage[caption=false]{subfig}
\usepackage{mathtools}
\usepackage{physics}
\usepackage{ulem}
\usepackage{dsfont}
\usepackage{tikz}
\usepackage[nameinlink,capitalize]{cleveref}

\usepackage{color}
\renewcommand{\arraystretch}{2.5}

\newcommand*\circled[1]{\tikz[baseline=(char.base)]{
		\node[shape=circle,draw,inner sep=2pt] (char) {#1};}}

\newtheorem{theorem}{Theorem}
\newtheorem*{theorem*}{Theorem}
\newtheorem{corollary}{Corollary}

\newtheorem{lemma}{Lemma}

\newtheorem{calculation}{Calculation}
\begin{document}
\title{Efficiencies of logical Bell measurements on CSS codes with static linear optics}
\date{\today}
\author{Frank Schmidt}
\email{fschmi@students.uni-mainz.de}
\affiliation{Institute of Physics, Johannes Gutenberg-Universit\"at Mainz, Staudingerweg 7, 55128 Mainz, Germany}
\author{Peter van Loock}
\email{loock@uni-mainz.de}
\affiliation{Institute of Physics, Johannes Gutenberg-Universit\"at Mainz, Staudingerweg 7, 55128 Mainz, Germany}
\begin{abstract}
	We show how the efficiency of a logical Bell measurement (BM) can be calculated for arbitrary CSS codes with the experimentally important constraint of using only transversal static linear-optical BMs on the physical single-photon qubit level. 
 For this purpose, we utilize the codes' description in terms of stabilizers in order to calculate general efficiencies for the loss-free case, but also for specific cases including photon loss.
	These efficiencies can be, for instance, used for obtaining transmission rates of all-optical quantum repeaters.
	In the loss-free case, we demonstrate that the important class of CSS codes with identical physical-qubit support for the two logical Pauli ($Z$ and $X$) operators can only achieve a logical BM efficiency of $\frac{1}{2}$ if one always employs the same ancilla-free static linear optical BMs on the physical level.
	We apply our methods to various CSS codes including two-dimensional planar color and planar surface codes.
	We then find that in many cases, the fixed use of the standard linear optical BM for all physical BMs is suboptimal and performing linear optical transformations before doing the standard linear optical BM (still without any ancillary photons and without any feedforward) can increase the efficiency enormously.
	In fact, using this generalization in the no-loss (or sometimes also in the low-loss) case allows us, on the one hand, to improve the logical BM efficiency of quantum parity codes compared to previously known results and, on the other hand, it also enables us to enhance the efficiency of two-dimensional planar color codes, whose efficiency is otherwise subject to the above $\frac{1}{2}$ bound, to arbitrarily close to unity.
\end{abstract}
\maketitle

\section{Introduction}

Quantum information theory provides many interesting and even practical applications like quantum computing, quantum (gate) teleportation, super-dense coding or quantum key distribution.
An essential resource is entanglement and a very important tool is a  Bell measurement (BM) which projects two qubits onto a basis of maximally entangled states.\\
Also in experimental realizations of quantum informational tasks there is typically the problem that one wants to protect and isolate qubits, but at the same time needs to address and manipulate them, thus preventing full isolation.
 The qubits then, at least to some extent, interact with their environment and lose their quantum coherence (for an uncontrollable environment).
This problem of decoherence can be tackled by using quantum error correcting codes that encode a logical qubit in multiple physical ones keeping the logical quantum information intact provided that not too many errors occured on the physical qubits, which is similar to methods from classical coding theory.

Quantum key distribution or more generally quantum communication via optical fibers suffers from photon loss for large distances.
 The direct quantum analogs of classical repeaters, which amplify the optical signals that propagate through a fiber channel at intermediate stations, are fundamentally impossible due to the quantum no-cloning theorem.
 Therefore, in the original proposal of a quantum repeater \cite{PhysRevLett.81.5932}, quantum information is not sent directly and instead copies of entangled states are initially distributed.
 These entangled states have to be processed by local operations and two-way classical communication, and also stored in quantum memories.
 Hence, transmission rates are fundamentally limited.
More recently, it was suggested to classify quantum repeaters into three generations \cite{repeatergenerations} and the most recent repeaters of the 3$^\text{rd}$ generation do not have a fundamental limit on the transmission rate anymore \cite{PhysRevLett.112.250501}. 

 The latter repeaters use quantum error correcting codes and perform error corrections after every few km's assisted by only one-way classical communication.
 A common scheme for this error correction is teleportation-based \cite{PhysRevA.71.042322}, where one performs quantum teleportation within the error correcting code, so a BM must be conducted within the code, which we refer to as a logical BM.
 An optical encoding is preferable for communication, since light allows for a high communication speed.
The most efficient and practical toolbox for manipulating quantum optical states is that of linear optics.
 It is well-known that it is impossible to perform a linear optical  BM with a greater efficiency than $\frac{1}{2}$ on optical dual-rail (DR) qubits (encoding a qubit into one photon occupying one of two, typically, polarization modes)  without making use of additional photons and without feedforward \cite{quant-ph/0007058}.
However, some effort \cite{Ewertqpcprl,EwertPRAqpc} has been made to study the benefits of quantum parity codes (QPC) \cite{qpcintro} for increasing the logical BM efficiency and robustness against photon loss.
These works show that it is possible to achieve an arbitrarily high logical BM efficiency for a sufficiently high number of code blocks in the absence of photon loss by using only static linear optics, i.e., without feedforward.
For sufficiently large and appropriately chosen block and qubit per block numbers, high and near-unit BM efficiencies are still obtainable in the presence of photon loss.
 Recently Lee et al. \cite{jeongfeedforward} showed that the efficiency of logical BMs on QPC can be improved further when making use of feedforward. Among the few existing quantum-error-correction-based repeater proposals based on linear optics, all \cite{Azuma2015,Pant2017,jeongfeedforward} except one \cite{Ewertqpcprl,EwertPRAqpc} depend on feedforward operations for the error correction steps. 
In the present paper, our focus remains strictly on a situation with static linear optics which appears to be beneficial in the view of the increase of the local operation times and the loss sensitivity of on-chip optical switching in integrated feedforward-based schemes.

Most of these works \cite{EwertPRAqpc,jeongfeedforward,Ewertqpcprl} only considered quantum parity codes, which are the generalized quantum versions of repetition codes, and it is well known that there exist many classical  and quantum codes which outperform repetition codes when considering ideal measurements.
Therefore we present a formalism that allows us to study the efficiencies of linear optical logical BMs for arbitrary CSS codes.
In previous works, the logical Bell states were decomposed into Bell states on the physical qubit level by making use of the repetition structure of the code.
Limited to the incomplete information available with the linear optical physical BMs, one can then check which measurement pattern can be unambiguously associated to a logical Bell state.
In our work, we make use of the stabilizer formulation of CSS codes in order to determine a linear system of equations for obtaining all possible measurement results in the physical Bell basis.
In order to see if an unambiguous identification of a logical Bell state is possible with linear optics, we simply check whether it is possible to fully cover the support of two logical Pauli operators with the existing information from the physical BMs.

We apply our systematic method to quantum parity codes \cite{qpcintro}, QPC($n,m$) with $n$ blocks and $m$ qubits per block, planar surface and planar color codes. 
It can be seen that the logical BM efficiency is very sensitive to linear optical transformations before the standard BM, still without any addition of ancillary photons or feedforward.
We also calculate the efficiencies in the presence of photon loss, but in this case we need to check for every combination of possible physical BM outcomes and erasure patterns whether it is possible to gain enough information to infer the logical Bell state.
Therefore, all our results including photon loss are based on python counting scripts, and we only obtain results for codes with a small number of physical qubits.
On the other hand, we also present results for more general codes, e.g. arbitrary planar surface codes, that rely on combinatorial arguments under the assumption that no photon loss occurs.
Hence we can see that the linear-optics constraint immensely increases the complexity of the analysis in comparison to the unconstrained BMs.

The paper is structured as follows.
 After introducing the general background of this work in \autoref{sec:background}, we show how a destructive logical BM can be performed by using only BMs on the physical qubit level in a transversal (and parallel) manner in \autoref{sec:logicalbm}.
 We further review linear optical BMs and show how the logical BM efficiency can be calculated for arbitrary CSS codes.
 We also make rigorous statements about the logical BM efficiency for wide classes of CSS codes.
 In \autoref{sec:codes} we apply our methods to two dimensional planar surface and color codes. We demonstrate that one can approach unit efficiency when using a special formation of different kinds of  linear optical BMs where otherwise more canonical formations result in a much lower efficiency no longer attaining unity.
Furthermore, we show how the logical BM efficiency of QPC($n,m$) can be increased compared to existing results \cite{Ewertqpcprl} without photon loss, still only based on static linear optics.
At the end, we compare small codes assuming photon loss and linear-optics constraint, and we find, for instance, that the Steane code is able to outperform direct transmission when allowing for additional unentangled ancillary single photons in the static linear-optics scheme.

\section{\label{sec:background}Background}

\subsection{Stabilizer and CSS codes}

We will give a short introduction to the stabilizer formalism \cite{gottesmanphd} which also includes our notation.
 Further details can be found in Ref. \cite{Nielsen:2011:QCQ:1972505}.
 The main idea of quantum error correcting codes typically is to encode some logical qubits in many physical qubits in order to be able to correct errors by redundancy.
The Pauli group of $N$ qubits is given by
\begin{eqnarray*}
	P_N= \{i^{\alpha}A_1 \otimes A_2 \otimes \dots  \otimes A_N|\alpha \in \{0,1,2,3\}, \\A_j \in \{I,X,Y,Z\}, j \in \{1,\dots,N\} \} \,,
\end{eqnarray*}
where $i$ is the imaginary unit, $I$ is the $2\times2$ identity matrix and $X,Y,Z$ are the $2\times2$ Pauli matrices.
 All elements of this group either commute or anticommute.
 A stabilizer code is defined by an abelian subgroup of the Pauli group.
 The elements of this group are called stabilizers and the code space of a stabilizer code is defined as the +1 eigenspace of all stabilizers.
  It is not necessary to take all stabilizer operators into account, since there is a subset of stabilizer generators that give every element of the stabilizer group by multiplication.
 In a quantum error correction scheme, these stabilizer generators will be measured, so that an error syndrome is obtained.
 This syndrome is used for estimating the most probable error and applying a suitable error correction.
 If the product of the actual error and the error correction operators is an element of the stabilizer group, the correction succeeds and if it commutes with all stabilizers but is not a stabilizer itself, the correction fails, because the logical information is changed.
A stabilizer code$[N,k,d]$ of code distance $d$ encoding $k$ logical qubits  into $N$ physical qubits has $2^{N-k}$ stabilizers and $N-k$ stabilizer generators.

A widely used technique to find such codes is the CSS code construction.
 It employs two classical linear codes $C_X$ and $C_Z$ \footnote{note that $C_X$ is called $C_2^\bot$  and $C_Z$ is called $C_1$ in Ref. \cite{Nielsen:2011:QCQ:1972505}} where $C_Z^\bot \subseteq C_X$ or equivalently $C_X^\bot \subseteq C_Z$ and $C^\bot$ refers to the dual code of code $C$.
 These conditions are needed in order to ensure that the stabilizer group is indeed abelian.
 One nice feature of these codes, CSS($C_X$,$C_Z$), is that the stabilizer generators only contain either $I$ and $X$ or $I$ and $Z$.
 This classification of the stabilizer generators means that $X$- and $Z$-errors can be corrected independently.
 The stabilizer generators can be constructed from the check matrices $H_X$ and $H_Z$ of the classical codes in the following way.
 For each row in $H_X$ add a stabilizer generator which contains an $X$-operator on the \textit{j}\textsuperscript{th} qubit if the \textit{j}\textsuperscript{th} entry of  this row is 1 and otherwise it should contain the identity operator on the \textit{j}\textsuperscript{th} qubit.
 The construction for the $Z$-stabilizers is analogous.
 Logical qubits and operators will be denoted by a bar, meaning, for instance, $\overline{X}$ flips the state of the logical qubit from $\ket{\overline{0}}$ to $\ket{\overline{1}}$.
Since the code is defined as a subspace, we are free to choose any basis for the logical qubits. As we are using CSS codes, we choose the basis in such a way that $\overline{X}$-operators are tensor products of $X$ and $I$ operators and $\overline{Z}$-operators are tensor products of $Z$ and $I$ operators.

In this paper, we will work with QPCs, planar surface codes and planar color codes.
A short review of these codes can be found in App. \ref{a:review}.

\subsection{Optical encoding, gates, and linear optics}
In our work, for the physical qubits, we consider DR encoded qubits which are defined by the presence of a single photon in one of two modes with the qubit basis states $\ket{0}$ ($\hat{=} \ket{01}$ in the Fock basis, i.e. $\ket{H}$ for polarization encoding) and $\ket{1}$  ($\hat{=} \ket{10}$ in the Fock basis, i.e. $\ket{V}$ for polarization encoding).
The polarization degree of freedom for the two modes excited by a photon is a very convenient choice of the DR encoding.
For DR qubits, arbitrary single-qubit operations can be performed with a general beam splitter (corresponding to polarization rotators for polarization encoding). 
For more qubits and modes, this corresponds to passive linear optics mixing annihilation operators of different modes linearly, which can be achieved in an experiment by sequentially applying beam splitters and phase shifters.
Annihilation operators $a$ are then transformed via
\begin{equation}
a_i^{\prime}=\sum_{j} U_{ij} a_j \,,
\end{equation}
 where $U_{ij}$ are elements of a unitary matrix $U$. Every matrix $U$ can be implemented via phase shifters and beam splitters \cite{reckdecomposition}. Due to the unitarity of $U$ the total photon number is preserved. However these linear-optical operations can only implement arbitrary unitary transformations acting on the annihilation operators and therefore all unitary transformations acting  on the Hilbert space of one photon in two modes (or even in more modes, but such a multiple-rail encoding will not be considered here), however, not all unitary operations acting on the multi-qubit Hilbert space of multiple photons, each encoded into two modes, can be performed. As a consequence it is impossible, for instance, to perform a CNOT gate on DR encoded qubits deterministically using only linear optics. Thus, also the canonical quantum circuit to realize a Bell-state measurement (see below) based on a two-qubit CNOT gate and a single-qubit Hadamard gate is impossible with linear optics. In the schemes considered in this paper, the matrix U is block-diagonal mixing only certain modes pairwise.

We define the four Bell states $\ket{\Phi^\pm},\ket{\Psi^\pm}$ as
\begin{equation}
	\begin{aligned}
		\ket{\Phi^\pm}=\frac{\ket{00}\pm\ket{11}}{\sqrt{2}}\quad &\hat{=} \quad ZZ=+1 \,,  &XX=\pm1 \,, \\
		\ket{\Psi^\pm}=\frac{\ket{10}\pm\ket{01}}{\sqrt{2}}\quad &\hat{=} \quad ZZ=-1 \,, &XX=\pm1 \,, 
	\end{aligned}
\end{equation} 
where the right side shows the correspondence to their eigenvalues of Pauli operators.
 A measurement on this Bell-state basis using static linear optics (i.e., without conditional linear-optical dynamics depending on measurement outcomes on a subset of modes) and without ancilla photons has an efficiency of at most $\frac{1}{2}$ \cite{quant-ph/0007058}.
It is then possible, for instance, to identify $\ket{\Psi^\pm}$ unambiguously by coupling both qubits with a 50/50 beam splitter.
In the context of quantum error correction, it was shown that it is impossible to deterministically transform a subspace of fixed photon number that is unable to correct photon loss into a code that can correct photon loss by applying only linear optics \cite{linearopticsconstrain}.

\subsection{Photon loss channel}

The effect of photon loss can be modeled by combining the mode of interest with an environmental mode initially in the vacuum state at a beam splitter.
 The idealization of using the vacuum state is reasonable, because $k_B T\ll \hbar \omega$ for frequencies in the optical domain at room temperature.
 The corresponding single-mode amplitude damping channel is given by 
\begin{equation}
\rho \rightarrow \sum_{k=0}^{\infty}\frac{(1-\eta)^k}{k!} \sqrt{\eta}^{a^{\dagger}a} a^k\rho a^{\dagger k}  \sqrt{\eta}^{a^{\dagger}a} .
\end{equation}
This non-unitary evolution can be obtained by tracing over the environmental mode after the beam splitter operation. The parameter $\eta$ represents the transmission probability for a photon, and $a$ is the bosonic annihilation operator of the relevant mode.
When we consider a DR encoded qubit and apply the amplitude damping channel with equal $\eta$ to both modes individually and independently, we obtain an overall channel evolution of
\begin{equation}
	\rho \rightarrow \eta \rho + (1-\eta)\ket{vac}\bra{vac} \,,
\end{equation}
which is also known as the erasure channel.
It maps a qubit state $\rho$ with a probability of (1-$\eta$) onto the two-mode vacuum state $\ket{vac}$ orthogonal to the qubit space (corresponding to an erasure) and otherwise the state is unchanged.
 An erasure pattern is undecodeable in the erasure channel iff  erasures fully cover at least one manifestation of the logical operators of the code \cite{1205.7036,7541599}.
Since $\overline{X}$ and $\overline{Z}$ anticommute, they must share at least one common qubit in every manifestation. 
If, for example, a manifestation of $\overline{X}$ is fully covered by erasures, then there is for each manifestation of $\overline{Z}$ at least one erased qubit which, however, would be needed in order to measure $\overline{Z}$. 

\subsection{Long-distance quantum communication}

Photon loss is a serious problem for long-range quantum communication, because the transmission probability $\eta$ decreases exponentially with distance, $\eta(L)=e^{-\frac{L}{L_{att}}}$ with an attenuation length $L_{att}$ and typically $L_{att}=22$km in an optical glass fiber.
Meanwhile upper bounds (so-called repeaterless bounds) \cite{PLOB,TGW} \footnote{A brief chronological summary of the different contributions
to the secret-key agreement capacity and the
corresponding bounds can be found in footnote [32] of
Ref. \cite{plobhistory} }
 on the point-to-point secret-key transmission rate were derived for a bosonic loss channel.
Furthermore, these upper bounds scale linearly with $\eta$ for $\eta\ll 1$.
In order to obtain a high quantum communication transmission rate for long distances, several schemes for quantum repeaters have been proposed.
The first proposals relied on quantum memories and entanglement purification requiring two-way classical communication.
Therefore, a fundamental limit on the transmission rate is given by the classical communication time \cite{repeatergenerations} for such quantum repeaters.
More recently, a quantum repeater based on only one-way communication was proposed (so-called 3\textsuperscript{rd} generation quantum repeater) \cite{PhysRevLett.112.250501}.
Since it only needs one-way communication, there is no fundamental limit on the transmission rate besides the local times for preparing and processing the encoded states.
This concept consists of splitting the overall distance $L$ into multiple smaller segments of $L_0$ like in a standard quantum repeater, however, with typically smaller $L_0$.

The steps of the protocol are as follows.
First, the quantum information gets encoded within an error correcting code and then this code is sent through a lossy channel of length $L_0$. Second, the error syndrome of the code is measured and a suitable correction is applied.
These steps are repeated at every station after each transmission over the length $L_0$ until the quantum information arrives at its destination.
 If all logical BMs were successful, the appropriate Pauli correction is applied, otherwise the received state is discarded.
Note that any proper loss code can, in principle, exceed the repeaterless bounds if we assume only channel loss with otherwise perfect measurements and corrections for arbitrary total distances. For a working loss code and correspondingly small repeater spacing $L_0$, the transmission probability of the logical qubit $\eta_{\text{log}}(L_0)$ is higher than for a physical qubit $\eta(L_0)$ and this effect increases exponentially with the number of repeater stations $\mathcal{N}$. Although direct transmission of DR qubits uses only two modes while the repeater uses 2$N$ modes independent of $L$ for a given code on $N$ DR qubits, the overall transmission rate per mode as relevant concerning the bounds of Refs. \cite{PLOB,TGW} can be improved via the repeater,
\begin{equation}
\frac{\eta_\text{log}^\mathcal{N}(L_0)}{2N}>\frac{\eta^\mathcal{N}(L_0)}{2}=\frac{\eta(L)}{2}\,,
\end{equation}
for some $L=\mathcal{N}L_0$, since $\left(\frac{\eta_\text{log}(L_0)}{\eta(L_0)}\right)^\mathcal{N}>N$ for some $\mathcal{N}$.

However, in general, the error correction steps require transversal, possibly nonlinear gate operations, and non-destructive syndrome measurements.
These experimentally demanding tasks can be simplified by making use of error correction by teleportation\cite{PhysRevA.71.042322,PhysRevLett.112.250501}.
This can be understood as a quantum teleportation protocol within an error correcting code where the error correction is successful if a logical Bell state can be identified correctly.
In this case, the logical qubit is recovered and can be passed on to the next station(see \autoref{fig:qr_scheme}).
Note that the above arguments for encoded transmissions with quantum error correction always beating direct, DR-encoded transmission if applied to error correction by teleportation at each repeater stations depend on the assumption of 100\% efficient BMs.
With standard linear optics this is impossible to achieve.
Therefore, it is not at all clear and highly non-trivial to decide whether and when an encoded transmission with linear-optics-based error correction steps at the repeater stations beats direct transmission of DR qubits.
One of our results will be that we can obtain such a logical BM for arbitrary stabilizer codes by performing BMs transversally on the physical qubit level. As a main result we will propose a method to calculate the logical BM efficiency for arbitrary CSS codes when using only linear-optical physical BMs in a static transversal manner.
Such a repeater for QPC($n,m$) and static linear optics was already proposed in Refs. \cite{Ewertqpcprl,EwertPRAqpc}. Here we generalize these results for arbitrary CSS codes.
\begin{figure}
\includegraphics[width=0.4\textwidth]{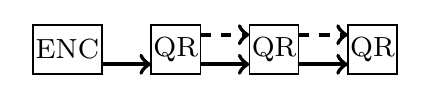}
\caption{A channel of distance $L$ is split into multiple segments of length $L_0=\frac{L}{\mathcal{N}}$ (here $\mathcal{N}=3$). The first station encodes the quantum information within an error correcting code and sends it through 2$N$ parallel bosonic-mode quantum channels (solid arrow) to the next station. A logical Bell state is already prepared in this station and a logical BM is performed on the incoming quantum code and on one half of the logical Bell state in order to teleport the quantum information into the other half of the prepared logical Bell state. This code is then again send to the next station (black solid arrow), while only $2\cdot k$ noiseless classical bits (dashed arrow) are transmitted in order to keep track of the required Pauli corrections due to the teleportation. This procedure is repeated until the last repeater station where an overall Pauli correction is applied and the code gets decoded or is otherwise consumed in some application.}
\label{fig:qr_scheme}
\end{figure}

\section{\label{sec:logicalbm}logical BM with stabilizer codes}
\subsection{Unconstrained physical BMs}
We will start showing that it is possible to decompose a destructive logical BM into multiple destructive physical BMs without the constraint of linear optics.
A destructive measurement is a measurement that gives us some information about the properties of a system, but this system may be destroyed by the measurement. 
Most photon detectors belong to this class of measurements, since they absorb the photons.
However, a destructive logical BM is sufficient for our purposes, because we only need the BM for performing quantum teleportation. 
Therefore, we are only interested in identifying the logical Bell state and do not care if the state is destroyed by the measurement.
\begin{theorem}
	A destructive logical BM on two copies of the same quantum error-correcting code can be performed by transversal physical BMs for every stabilizer code.
\end{theorem}
\begin{proof}
	A BM of two qubits is equivalent to a measurement of $XX$ and $ZZ$, since Bell states are unique and simultaneous eigenstates of these two operators.\\
	Let $N$ be the number of physical qubits in the code and $k$ the number of encoded logical
	qubits. 
	A logical BM corresponds to a measurement of the operators $\overline{X}_{1,t}\overline{X}_{2,t}$ and  $\overline{Z}_{1,t}\overline{Z}_{2,t}$ for all $t\in\{1,\dots,k\}$  on two quantum error-correcting codes (indices 1 and 2 indicate the two codes).
	 Now we number the physical qubits of both codes from 1 to $N$  using the same numbering.
	Then we measure $X_{1,o}X_{2,o}$ and $Z_{1,o}Z_{2,o}$ for all $o\in \{1,\dots,N\}$.
	 Notice that all logical operators have in common that whenever a Pauli operator acts on a physical qubit of code 1, there also acts an identical Pauli operator on the corresponding physical qubit in code 2.
	Therefore, all physical $X_{1,o}X_{2,o}$ and $Z_{1,o}Z_{2,o}$ commute with themselves and with the logical operators $\overline{X}_{1,t}\overline{X}_{2,t}$ and  $\overline{Z}_{1,t}\overline{Z}_{2,t}$.
 Thus, all of these operators are diagonalizable simultaneously and  measuring all $X_{1,o}X_{2,o}$ and $Z_{1,o}Z_{2,o}$ gives enough information for discriminating all logical Bell states.
 The measurement outcome of $\overline{X}_{1,t}\overline{X}_{2,t}$ can be obtained by multiplying the measurement results of $X_{1,o}X_{2,o}$, $Z_{1,o}Z_{2,o}$ and $Y_{1,o}Y_{2,o}=X_{1,o}X_{2,o}Z_{1,o}Z_{2,o}$ along the support of $\overline{X}_{1,t}\overline{X}_{2,t}$ (for CSS codes only $X_{1,o}X_{2,o}$ need to be multiplied since $\overline{X}$ has only X-operators in its support).
 The case of $\overline{Z}_{1,t}\overline{Z}_{2,t}$ works analogously.
\end{proof}
\noindent Note that the above method differs from a direct measurement of the logical operators in the way that it may destroy the code while the direct measurement of the logical operators ensures to preserve the code.\newline
Such a transversal BM-based error correction by teleportation is well suited for a linear optical implementation of erasure (loss) correction, since we do not need to perform a quantum non-demolition measurement in order to detect photon loss.
The logical BM also contains the relevant photon number information, as the total number of photons entering the logical Bell-state analyzer will be preserved by linear optics.

It is conceptually rather easy to calculate the transmission
probability, i.e. the probability of a successful logical BM after
applying the loss channel on the physical qubits.
It can be obtained obtained by counting all decodable erasure patterns and weighing them with their probability of occurrence. 

For our purposes, it will be useful to assume that photon loss occurs with a transmission probability of $\eta_1$  on the photons of code 1 and with $\eta_2$ on the photons of code 2.
 It is also useful to allow different transmission probabilities for the two codes, because in the case of a quantum repeater (see \autoref{fig:qr_scheme}), one code is typically sent through an optical fiber, while the other code is prepared and consumed locally at each station.
 Thus,  the second code will have much less losses on average than the first code.
We denote the probability that both photons that are needed for a physical BM are not lost by $\tilde{\eta}=\eta_1 \eta_2$.

As a result, the transmission probability of a logical qubit of an [$N,k,d$]-code after one error correction step is given by
\begin{equation}
\label{eq:trans}
	\eta_{log}=\sum_{j=0}^{N} e_j\tilde{\eta}^{N-j}\left(1-\tilde{\eta}\right)^j \,,
\end{equation}
where $e_j$ is the number of decodable erasure patterns containing $j$ erasures and depending on the code, and where we assume that every physical BM succeeds.

As an example, consider for each code QPC(2,2) (see App. \ref{a:review}), i.e., $N=n\cdot m=4$, $k=1$, so that we have 4-1=3 stabilizer generators, $\{XXXX,ZZII,IIZZ\}$, and the logical operators $\overline{X}=XXII=IIXX$ and $\overline{Z}=ZIZI=ZIIZ=IZIZ=IZZI$, where the qubits can be arranged in two blocks, each containing two qubits (here the first two qubits for block 1 and the second two qubits for block 2). As it can be seen in \autoref{fig:qpc22}, the corresponding qubits in both codes are measured in their physical Bell bases, and there must be at least one intact qubit pair per block in order to measure $\overline{ZZ}$ and one whole block must be intact in order to measure $\overline{XX}$. This condition is basically the decodability condition of a single code in the erasure channel, where no logical Pauli operator should be fully covered by erasures (at least one qubit per block needs to be intact, such that no $\overline{X}$ is fully covered by erasures and a whole block needs to remain intact, such that no $\overline{Z}$ may be fully covered by erasures). 
\begin{figure}
	\subfloat[\label{sfig:qpc22figure}]{
		\includegraphics[width=0.2\textwidth]{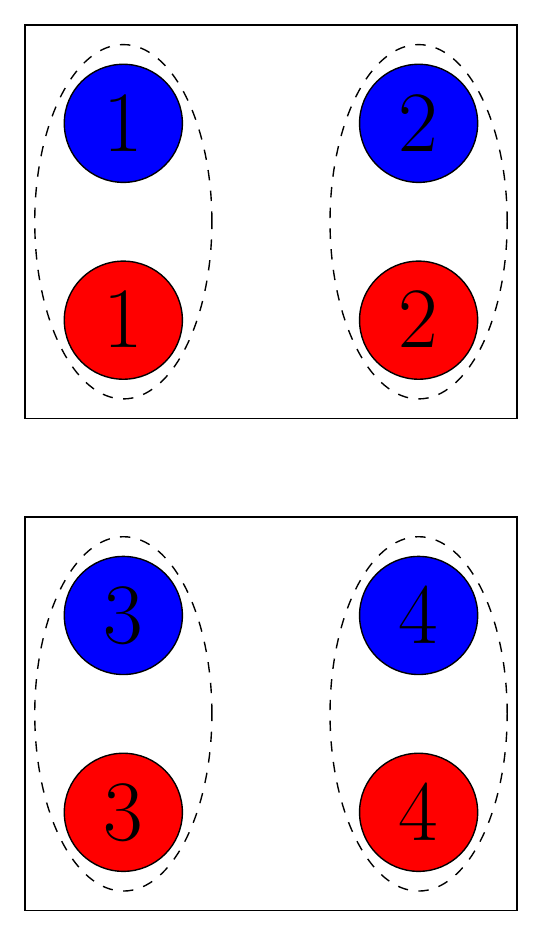}
	}
	\subfloat[\label{sfig:qpc22withloss}]{
		\includegraphics[width=0.2\textwidth]{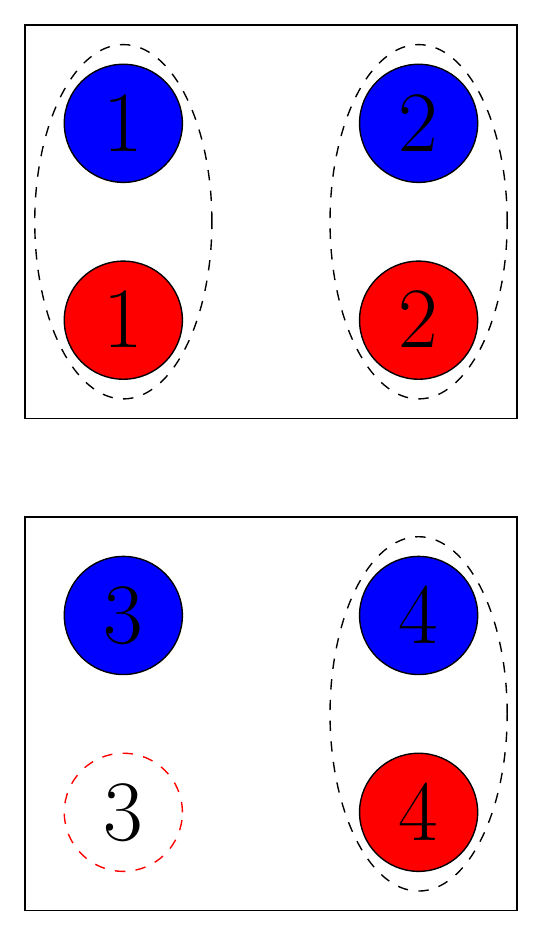}
	}\\
	\subfloat[\label{sfig:qpc22figure_lin}]{
		\includegraphics[width=0.2\textwidth]{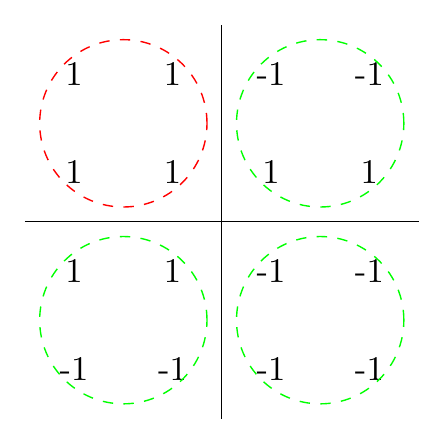}
	}
	\subfloat[\label{sfig:qpc22figurewithloss_lin}]{
		\includegraphics[width=0.2\textwidth]{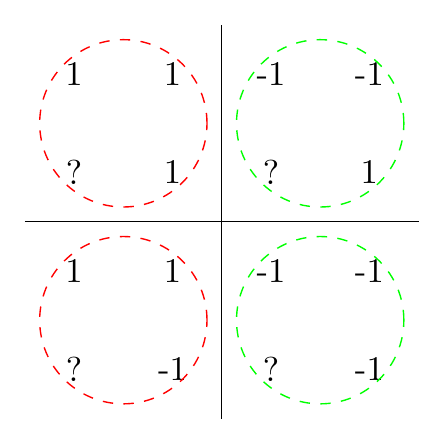}
	}
\caption{Logical BM on QPC(2,2). The blue circles denote qubits of the first copy of QPC(2,2), while the red circles denote qubits of the second copy of QPC(2,2). Physical BMs are denoted by dashed ellipses and are performed on corresponding qubit pairs. The two blocks of each code are indicated by the black boxes. (a) Unconstrained loss-free case. Since all qubit pairs are intact, a logical BM always succeeds using unconstrained physical BMs. (b) Unconstraint lossy case.  Photon 3 of the second code is lost, such that we cannot perform a BM on the third qubit pair. However, qubit pairs 1 and 2 give us information about $\overline{XX}$ and e.g. qubit pairs 1 and 4 still give us information about $\overline{ZZ}$. As a consequence we are still able to perform the logical BM. In (c) and (d) we consider the linear-optics constraint. This means for the physical BMs we always obtain the $ZZ$ measurement result, but we only obtain the XX measurement result if the $ZZ$ measurement yielded -1. We show the possible $ZZ$ measurement results and mark those results where a logical BM is possible in green. (c) No photons are lost and only one of the four possibilities is a failure. (d) Loss of photon 3 in code 2 and now two of the four possibilities are failures.}
\label{fig:qpc22}
\end{figure}

\subsection{Linear optical physical BMs}

Unfortunately, it is impossible to perform physical BMs  deterministically by using linear optics (not even including feedforward and ancilla photons  \cite{quant-ph/9809063}, except for near-deterministic BMs \cite{KLM}).
However, one may consider a probabilistic measurement which sometimes can identify Bell states unambiguously.
 For this kind of measurement it was shown \cite{quant-ph/0007058} that the BM efficiency is bounded by 50\% if one assumes an equiprobable ensemble of Bell states and only allows for static and ancilla-free linear optics. 
There exists a very simple linear optical scheme for such a BM saturating the bound, which is also known as the Innsbruck scheme \cite{innsbruck1,innsbruck2,innsbruck3}  or the standard BM. It allows for identifying two of the four Bell states, thus achieving a 50\% BM efficiency only by using a 50/50 beam splitter.
Since an arbitrary single-qubit operation on DR qubits can be performed by linear optics, it is also possible to identify two arbitrary Bell states by simply performing linear optical operations before the standard BM.
The standard BM (for polarization encoding) employs a phase-free 50:50 beam splitter that mixes the two qubits.
Then the resulting two outputs are each split by a polarization beam splitter, so that horizontal and vertical polarizations get spatially separated.
This is followed by a measurement in the photon number basis and those cases where only one photon per detector is detected lead to a successful identification.

When allowing for additional ancilla photons it is possible to increase the BM efficiency \cite{Grice,EwertAdvancedBM,grosshans_ancilla} and it is even possible to achieve an arbitrarily high efficiency for sufficiently many extra photons.
These measurements can also identify two of the four Bell states deterministically, while the other two states can still be identified with a probability $0\leq p_{\text{adv}}<1$ .

Linear optical BMs can be put into two categories depending on whether their success probability varies with the input Bell state. BMs which can identify all four Bell states with the same probability can be analyzed easily.
In this case failed BMs can be also seen as additional erasures.
Thus, we can use the results of ideal BMs and replace the transmission probability $\tilde{\eta}$  by $\tilde{\eta} \cdot p_{BM}$ where $p_{BM}$ is the success probability of the physical BM.
Recall that for an equiprobable ensemble of the four Bell states a linear optical BM without ancilla photons cannot achieve a higher efficiency than $\frac{1}{2}$ \cite{quant-ph/0007058}. 
Therefore, it is impossible due to the no-cloning theorem that a quantum code gives a higher transmission rate, when using a physical BM whose efficiency is independent of the input state (see App. \ref{a:nocloning}).\newline
An example of a measurement of the other, useful  category is the standard linear optical BM, which identifies the Bell states $\ket{\psi^\pm}$ unambiguously.
It is insightful to write these measurements in terms of eigenstates of Pauli operators.
For example, $\ket{\psi^+}$ is an eigenstate of $ZZ$ with eigenvalue -1 and it is an eigenstate of $XX$ with eigenvalue +1. On the other hand, $\ket{\psi^-}$ is an eigenstate of $ZZ$ with eigenvalue -1 and it is an eigenstate of $XX$ with eigenvalue -1.
The states that cannot be discriminated, $\ket{\Phi^\pm}$, are $ZZ$ eigenstates with eigenvalue +1.
Therefore, we obtain  $ZZ$-information in any case, but we only get further information about $XX$ provided the $ZZ$ measurement gives -1 as the measurement result.
However, this standard linear optical BM is not the only one that can detect two of the four Bell states.
There are $\binom{4}{2}=6$ different measurements possible and these correspond to the $3\cdot 2$ cases where one has three possibilities for the guaranteed information (either $XX$, $YY$ or $ZZ$) and one gets the additional information if the guaranteed result is either +1 or -1.
These five other BMs can be obtained by applying single-qubit transformations (which can be done via linear optics for the DR encoding) before performing the standard BM.
The most general linear optical BM without ancilla photons and without feedforward can identify the four Bell states with probabilities $p_1,\cdots,p_4$ satisfying $\sum_{j=1}^4p_j\leq2$.

We need to decompose the logical Bell states into Bell states on the physical qubit level for calculating the logical BM efficiency when using BMs with state-dependent efficiency.
We will assume for the remainder that all logical Bell states appear with equal probability.
The stabilizer conditions of a CSS code will give us a set of multiplicative equations for possible BM results.\newline
We will use a simple transformation that allows us to solve this set of equations with usual methods of linear algebra over a finite field:
\begin{equation}
	\label{eq:transformation}
	\begin{aligned}
	+1&\longleftrightarrow 0\quad\,,\\
	-1&\longleftrightarrow 1\quad\,,\\
	\cdot &\longleftrightarrow +  \quad.
	\end{aligned}
\end{equation}
Thus, from now on, we refer to a measurement like the standard linear optical BM as a $ZZ$=1 BM, since we always get $ZZ$-information, but we only get the additional $XX$-information if the $ZZ$ measurement yields the eigenvalue $-1$ (if the ''good'' eigenvalue was +1, we would refer to the BM as $ZZ$=0 BM).\newline
The idea of the efficiency calculation is to make a list of all possible BM results and count the number of those results where enough information for an identification of a logical Bell state can be obtained via the linear optical BMs.
 First we assume that no photon losses occur and later we will generalize our formalism to include photon loss.
For a given logical Bell state the possible physical BM results are restricted due to the stabilizer conditions of both codes involved.
We shall address this in the following.\newline
\subsubsection{Loss-free case}
Since we are working with CSS [$N,k,d$] codes, we have stabilizer generators  $g_{X_{h,1}},\dots,g_{X_{h,l'}}$ corresponding to $X$-type stabilizers of the \textit{h}\textsuperscript{th} code ($h\in \{1,2\}$) and $g_{Z_{h,1}},\dots,g_{Z_{h,l''}}$ corresponding to $Z$-type stabilizers of the \textit{h}\textsuperscript{th} code ($l'+l''=N-k$).
Thus, the subspace of the two (identical) quantum codes can be described by the set of independent stabilizer generators, 
\begin{equation*}
	\begin{aligned}
		 \{g_{X_{1,1}},\dots,g_{X_{1,l'}},g_{X_{2,1}},\dots,g_{X_{2,l'}},\\g_{Z_{1,1}},\dots,g_{Z_{1,l''}},g_{Z_{2,1}},\dots g_{Z_{2,l''}}\}\quad\,,
	\end{aligned}
\end{equation*}
encoding $2k$ logical qubits into 2$N$ physical qubits.
However, note that for every stabilizer generator there exists an $XX$ or $ZZ$ operator that anticommutes with the stabilizer generator.
Thus, it is more useful to use 
\begin{equation*}
\begin{aligned}
\{g_{X_{1,1}},\dots,g_{X_{1,l'}},g_{X_{1,1}}g_{X_{2,1}},\dots,g_{X_{1,l'}}g_{X_{2,l'}},\\g_{Z_{1,1}},\dots,g_{Z_{1,l''}},g_{Z_{1,1}}g_{Z_{2,1}},\dots,g_{Z_{1,l''}} g_{Z_{2,l''}}\}
\end{aligned}
\end{equation*}
as a set of independent stabilizer generators, because half of these stabilizer generators commute with all $XX$ and $ZZ$ operators and hence their eigenvalues are preserved after a BM on the qubit level. 
When considering the measurement of $XX_j$  (abbreviation of $X_{1,j}X_{2,j}$ ) for $j\in \{1,\dots,N\}$ there exists a $Z$-type stabilizer generator that anticommutes with $XX_j$ and thus it can yield 0 and 1 as a result with equal probability.
However, a string of all $XX$ measurement results ($XX_1,\dots,XX_N$) cannot take any value in $\mathbb{Z}_2^N$, since the stabilizer conditions of the generators $g_{X_{1,1}}g_{X_{2,1}},\dots,g_{X_{1,l'}}g_{X_{2,l'}}$ need to be fulfilled.
Fulfilling these stabilizer conditions corresponds to being an element of the kernel of $H_X$ when using the above transformation in \cref{eq:transformation}.
Thus, we can obtain a list of all possible $XX$ measurement outcomes by calculating the codewords of the classical linear code $C_X$, while we can obtain a list of all possible $ZZ$ measurement outcomes by calculating the codewords of $C_Z$.\newline
For each stabilizer generator in $\{g_{Z_{1,1}},\dots,g_{Z_{1,l''}}\}$ there exists an $XX$ operator that anticommutes at least with this stabilizer generator. If this $XX$ operator is measured, it gives 0 or 1 as a measurement result with probability $\frac{1}{2}$.
When also considering that $\overline{X}_{1,m}\overline{X}_{2,m}$ can give two different measurement outcomes with equal probability, since we perform measurements on an encoded mixture of equiprobable Bell states, for each $m \in \{1,\dots,k\}$, we see that we obtain $2^{l''+k}=2^{N-l'}$ different equiprobable measurement patterns of $XX$ operators. 
Note that we already showed that every pattern of $XX$ is restricted to be a codeword of $C_X$ and the subspace $C_X$ consists of $2^{N-l'}$ codewords.
Therefore, every codeword of $C_X$ occurs with equal probability. 
Similarly, one can show that also all codewords of $C_Z$ occur with equal probability. \newline
Since we now have a list of all possible physical BM results, it only remains to be checked how many of those give us information about $\overline{XX}$ and $\overline{ZZ}$ when restricted to the information available via linear optical BMs.
In general, we have not found a more efficient method for checking than counting all possible patterns,
though simplifications can be made if one considers using either only BMs with guaranteed $ZZ$-information or only BMs with guaranteed $XX$-information.
 In these cases, it is not necessary to look at all the codewords of $C_Z$ and $C_X$.
Instead one can reduce the number of codewords that has to be checked, because e.g. the success of guaranteed $ZZ$-information measurements is independent of the $C_X$-codewords and hence one only needs to consider the $C_Z$-codewords. \newline
As an example, let us again consider QPC(2,2), but now we do not perform unconstrained physical BMs, but $ZZ$=1 BMs.
Therefore  the efficiency only depends on the $ZZ$ measurement results. The four possible $ZZ$ results can be seen in \autoref{fig:qpc22} (c) and (d). Recall that we need to obtain at least one $ZZ$ result per block and one block where we obtain the $XX$ result for each qubit pair. We obtain $ZZ$-information every time when a physical BM can be performed on two DR rail qubits (none of which were lost), but we only obtain $XX$-information when the $ZZ$ measurement yielded result 1. Therefore, obtaining $\overline{XX}$-information is much more problematic than obtaining $\overline{ZZ}$-information. We can see in \autoref{fig:qpc22} (c) that three of four possible $ZZ$ result combinations allow for a successful logical BM, while we can see in \autoref{fig:qpc22} (d) that only two of four possible $ZZ$ result combinations allow for a successful logical BM when one qubit pair is (partially) erased. In comparison, the logical BM can be even performed deterministically with one erased qubit pair when using unconstrained physical BMs.\newline
It is also straightforward to include enhanced linear optical BMs like those of Refs. \cite{Grice,EwertAdvancedBM,grosshans_ancilla}, which identify two logical Bell states deterministically and  identify the other two with equal probability $p_{adv}$ with the help of ancilla photons.
In order to calculate the logical BM efficiency, we go through the same steps as for the case with standard ancilla-free linear optical BMs, but now we also look at those codewords that previously did not reveal enough information for a successful logical BM. 
We have to take into account that every physical BM that fails when using non-enhanced linear optical BMs has a probability of $p_{adv}$ to be successful.
Thus, we have to create success/failure patterns for all physical BMs that do not succeed deterministically and we have to count those patterns that allow for a successful logical BM with weights corresponding to their probability of occurrence depending on $p_{adv}$.\newline

Stabilizers and logical operators often do have a geometrical interpretation and it is useful to apply this also to the BM outcomes.

\begin{lemma}\label{lem:enum}
Z-codewords can be associated with the support of X stabilizer generators and $\overline{X}$-operators and vice versa (replacing Z and X by X and Z, respectively) for any CSS($C_X$,$C_Z$)-code.
\end{lemma}
\begin{proof}
Assume that the code uses $N$ physical qubits, $k$ logical qubits and $j$ $Z$-stabilizer generators.
 The space of these $Z$-codewords is an $(N-)j$-dimensional subspace of $\mathbb{Z}_2^N$. For these codewords there is a one-to-one correspondence with strings of $N$ bits, where the $h$\textsuperscript{th} bit of such a string is 1 if $X_h$ is used in a given $X$-stabilizer or $\overline{X}$-operator.
  $X$-stabilizers and $\overline{X}$-operators  commute with all $Z$-stabilizers and thus their support corresponds to codewords, because $X$-stabilizer and $\overline{X}$-operators  contain only $X$ and $I$-operators  in CSS codes.
 There are $N-j-k$ independent $X$-stabilizer and $k$ $\overline{X}$-operators and thus we have $N-j$ linear independent codewords, which form a basis for all codewords of $C_Z$, since the space is $(N-j)$-dimensional.
\end{proof}

As a small example let us now reproduce the no-loss QPC($n,m$)-efficiency from \cite{Ewertqpcprl,jeongPRL}, where only $ZZ$=1 measurements are used, by applying our new method.
QPC($n,m$) encodes $k=1$ logical qubit, see App. \ref{a:review}.
We will begin with the special case of $m$=1 qubit per block and the stabilizer generators for this case are given by $X_{o,1}X_{o+1,1}$ with $o$ $\in \{1,\dots,n-1\}$ (the first index labels the block number and the second index labels the qubits within a block; the numbering of the two codes will be omitted). There are no $Z$-type stabilizer generators for the case $m$=1 (corresponding to $j=0$ in the proof of \autoref{lem:enum}) By using $X$ we refer to one code, while we refer to both codes by using $XX$.  $\overline{Z} $ is given by $\prod_{o=1}^n Z_{o,1}$ and $\overline{X}$ is given by  $X$ on an arbitrary qubit.\newline
 The $\overline{ZZ}$ information is already given, since we get the $ZZ$-information for each qubit pair.
Then we also need the $\overline{XX}$ information for a successful identification of the logical Bell state. Since we only need the $XX$-information from one qubit pair, the only possible case where the logical BM fails is when all $ZZ$ measurements yield 0 which corresponds to 1 out  of $2^n$ cases, because we have $n$-1 $X$-stabilizer generators and 1 $\overline{X}$-operator (using \autoref{lem:enum} as a one-to-one correspondence between $Z$-codewords and $X$-type operators).
 Increasing $m>1$ does not change the efficiency, because the number of $X$-operators stays the same and thus the number of codewords in $C_Z$ remains the same due to \autoref{lem:enum}. The weight of the $X$ operators changes too, but this does not matter, since, like in the $m$=1 case, any $X$ stabilizer generator fully covers an $\overline{X}$, such that the combinatorial argument remains the same.
Therefore, we obtain $1-2^{-n}$ as the logical BM efficiency of QPC($n,m$) in the absence of photon loss.
 For the specific example of QPC(2,2), see \autoref{fig:qpc22} where still "0" is "+1" and "1" is "-1" according to \cref{eq:transformation}.\newline

After describing our simple formalism of performing logical BMs with linear optics, we are able to make some general statements about the logical BM efficiency of CSS codes.

\begin{theorem}\label{thm:count}
Using either  $XX$=1-BMs or $ZZ$=1-BMs for all qubit pairs (and not a combination of both) gives a logical BM efficiency of at least $\frac{1}{2}$ in the loss-free case for an $[N,1,d]$-CSS($C_X$,$C_Z$)-code
\end{theorem}
\begin{proof}
Without loss of generality we are assuming that we use $ZZ$=1-BMs. 
We are looking for possible $ZZ$ measurement outcomes, which fulfill the $j$ $Z$-stabilizer conditions and thus are codewords of $C_Z$.
We make use of \autoref{lem:enum}. Combinations corresponding to $\overline{X}$ times a product of $X$-stabilizers have "1"s along an $\overline{X}$-operator by definition and therefore the BMs give additional information along the support of $\overline{X}$ such that we can identify the logical Bell state. 
Thus, one obtains an efficiency of at least $\frac{1}{2}$.
\end{proof}
 This theorem cannot be directly generalized to codes encoding more than one logical qubit, because it is possible that $\overline{X}_1$ and $\overline{X}_2$ have some physical qubits in common. 
 Note that the assumption that the guaranteed information needs to be 1 in order to get full information cannot be omitted, since, for example, the efficiency for the planar surface code(2,2) using only $ZZ$=0-BMs with $p_\text{adv}=0$ is smaller than $\frac{1}{2}$ .\\

 \begin{theorem}\label{thm:xzgeo}
The logical BM efficiency of a CSS code cannot exceed $\frac{1}{2}$ using linear optics with identical guaranteed-information BMs for all qubit pairs and $p_{\text{adv}}=0$ if support($\overline{X}$)= support($\overline{Z}$).
\end{theorem}
\begin{proof}
Without loss of generality we are using $ZZ$=1-BMs. 
The number of physical qubits in support($\overline{X}$)=support($\overline{Z}$) is odd, as otherwise it is impossible that $\overline{X}$$\overline{Z}$=-$\overline{Z}$$\overline{X}$. 
We demand to measure only "1"s as $ZZ$ results along the path of an $\overline{X}$ in order to resolve $\overline{XX}$, but by assumption this means along the path of a $\overline{Z}$. 
Adding 1 for an odd number of times gives 1(mod 2) and this means that the ability of doing a successful logical BM is only given if $\overline{ZZ}$=1, which is only the case for half of the logical Bell states. 
Thus, the averaged logical BM efficiency cannot exceed $\frac{1}{2}$.
\end{proof}
In other words, in this case, the well-known bound of $\frac{1}{2}$ for the BM efficiency of a standard physical linear-optics BM also applies to the logical (transversally performed) linear optics BM efficiencies.   
However, it is possible to circumvent \autoref{thm:xzgeo} by using a combination of different physical BMs, e.g. a BM with guaranteed $XX$-information on some physical qubits and a BM with guaranteed $ZZ$-information on the  other physical qubits, and we will see in the following subsection that this additional degree of freedom can have a huge impact on the logical BM efficiency.\\

\begin{corollary}\label{col:xzgeo}
Using an $[N,1,d]$-CSS($C_X$,$C_Z$)-code with $C_X=C_Z$, it is impossible to exceed a logical BM efficiency of $\frac{1}{2}$  using linear optics with identical guaranteed-information BMs for all qubit pairs and $p_{\text{adv}}=0$ .
\end{corollary}
\begin{proof}
The support of $\overline{Z}$-operators is in $C_X\backslash C_Z^{\bot}$ and for $\overline{X}$-operators it is in $C_Z\backslash C_X^{\bot}$.
 By assumption $C_X=C_Z$ and the dimension of  $C_X/ C_Z^{\bot}$=1.
 Thus,  support($\overline{X}$)= support($\overline{Z}$), because it is not possible to find two different bases over the field $\mathbb{Z}_2$ for a 1-dimensional space.
 Using \autoref{thm:xzgeo} the result follows.
\end{proof}

This corollary includes a large class of codes such as all two-dimensional planar color codes.

\begin{theorem}\label{thm:csslinearineq}
A CSS code with $j$ $Z$-stabilizer generators and $N$ physical qubits cannot achieve a better logical BM efficiency than 
\begin{equation}
1-2^{-\left(N-j\right)}\,,
\end{equation}
if only $ZZ$=1-BMs are used and $p_{\text{adv}}=0$.
\end{theorem}
\begin{proof}
The codewords of $C_Z$ constrain the possible outcomes of the physical $ZZ$-measurements.
 Every linear code has a neutral element, which means that every $ZZ$ measurement yields 0, such that no $\overline{XX}$-information is available.
 Counting the number of all possible physical equiprobable $ZZ$ outcomes gives $2^{N-j}$. 
\end{proof}
Notice that QPC($n,m$) uses $N=n\cdot m$ qubits and $j=n\cdot(m-1)$ $Z$-stabilizer generators and therefore \autoref{thm:csslinearineq} bounds the logical BM efficiency to $1-2^{-n}$ which can be achieved by transversal, static linear optics with $p_{\text{adv}}=0$ \cite{Ewertqpcprl}.

Up to now we only considered linear optical BMs that were either guaranteed-information BMs or BMs with an efficiency independent of the input Bell state.
Let us now also consider the most general ancilla-, feedforward-free linear optical BM compatible with the upper bound of Ref. \cite{quant-ph/0007058}, which identifies the four Bell states unambiguously with arbitrary probabilities $p_1,\dots,p_4$ where $\sum_{j=1}^4p_j\leq2$ (the two cases discussed so far only include $p_1=p_2=p_3=p_4=\frac{1}{2}$ and some pair $p_k=p_l=1$ while another pair $p_r=p_s=0$). 
In the following, we will derive an upper bound on the logical BM efficiency based on these ancilla-free linear optical BMs.
In order to do so, let us assume that we do not perform the physical BMs by ourselves. Instead we let another trustworthy party called Charlie, who claims to be able to perform such a general ancilla-free linear-optical BM (in principle, we could send him random Bell states and check whether his measurement results follow the desired statistics), do the physical BMs. We will then use Charlie's measurement outcomes for identifing the logical Bell state.
However, Charlie's protocol simply consists of performing one of the six guaranteed-information BMs according to a probability distribution that ensures that our expected statistics are fulfilled.
It is intuitively clear that such a probability distribution exists, but in App. \ref{app:probdistri} we explicitly show how Charlie could construct such a probability distribution.
Next he will inform us about his measurement outcomes, but he will not tell us which guaranteed-information BMs were used by him. Since Charlie has all the information that we got and, additionally, he also knows which guaranteed-information BMs were used, his averaged logical BM efficiency is therefore at least as large as our logical BM efficiency using general ancilla-free linear optical BMs.
Furthermore, Charlie's averaged logical BM efficiency is given by the logical BM efficiency of fixed guaranteed-information formations weighted  by their probability of occurrence according to Charlie's chosen probability distribution.
Therefore, his averaged  logical BM efficiency is at most as large as the logical BM efficiency of his optimal guaranteed-information formation.

Note that this argument does not rely on assumptions about $\eta$ and therefore we can restrict ourselves to the set of guaranteed-information BMs when looking for optimal logical BM efficiencies. However, the choice of the optimal guaranteed-information formation may depend on $\eta$.
In addition, one can also use a similar argument when considering general advanced linear-optical BMs (which are only forbidden to identify all four Bell states deterministically \cite{quant-ph/9809063}), although one needs to be cautious what desired statistics can be obtained by using probabilistically advanced guaranteed-information BMs.
 For example, Charlie could not obtain a BM that always identifies two specific Bell states, identifies one state with probability $50\%$, and never identifies the remaining Bell state by using his strategy.
However, it is always possible to obtain a BM whose efficiency is independent of the state by using advanced $ZZ=0$ and $ZZ=1$, each with probability $\frac{1}{2}$.
For such a protocol, the whole linear-optics constraint simplifies to an additional erasure channel where the transmission is given by the physical BM efficiency $p_{BM}$.
As a consequence, analyses of arbitrary quantum codes considering erasures and possibly additional errors can be used as a lower bound on the logical BM efficiencies using static linear optics.
By employing advanced guaranteed-information BMs as described in Ref. \cite{EwertAdvancedBM}, it is possible to obtain $p_{BM}=\frac{3}{4}$ without using entangled ancillae. 
For example, it can therefore be immediately seen that for $\eta>\frac{2}{3}$ a logical BM efficiency of approximately unity can be obtained by using large-distance toric codes \cite{PhysRevA.81.022317}.

\subsubsection{\label{sec:codes}Application of loss-free results to specific codes}
We will discuss the logical BM efficiencies for planar surface and planar color codes and how we can influence these efficiencies by using combinations of  different linear optical physical BMs in a transversal and static manner.
Detailed calculations of the efficiencies can be found in App. \ref{sec:calculations}.
We emphasize that all general results for codes of arbitrary size apply only to the loss-free case.
 Results including loss have been obtained by a small script that simply tests all codewords and erasure patterns eventually leading to the coefficients $e_j$ in \cref{eq:trans} \footnote{The script can be found at \url{https://github.com/schmidtfrk/CSSlogBM}.}.
 Therefore, the computation time grows exponentially with the number of qubits such that we only obtain results for small codes like, for example, QPC(4,2) \footnote{Note that the logical BM efficiency of general QPC($n,m$) is known when considering loss and standard $ZZ$=1 BMs \cite{Ewertqpcprl,EwertPRAqpc}. However, here we also consider combinations of different guaranteed-information BMs, breaking the block structure of the gained information and thus rendering the technique of Refs. \cite{Ewertqpcprl,EwertPRAqpc} inapplicable.}, surface(3,2), and the Steane code, as will be discussed in the subsequent section.

One can use combinatorics for obtaining the no-loss efficiency of logical BMs within planar color or planar surface codes.
 If we restrict ourselves to using only $ZZ$=1 measurements, we can see immediately that the logical BM no-loss efficiency of planar color codes in the no-loss case is given by $\frac{1}{2}$ due to \autoref{col:xzgeo} and \autoref{thm:count}.
Therefore, we decided to study the efficiency of planar surface codes which are closely related to planar color codes, but there is no no-go result for achieving a high efficiency when considering only standard linear optical BMs.
 The non-trivial part of counting those physical BM patterns of surface($n,m$) that allow for an identification of the logical Bell states is equivalent to counting all solutions in the following combinatorial problem.
Consider a board with $(n- 1)\cdot m$ squares with values 0 and 1 where one wants to cross the board (along the side with distance $m$) and one can only walk on squares of 1 and to adjacent squares of also 1 (moving diagonally is allowed). 
With this geometrical interpretation it is easy to see that increasing $m$ will result in a decrease of the no-loss efficiency, because in addition to the previous rows one must be able to cross the extra rows.
 This is already a first difference to QPC($n,m$) where it is impossible to decrease the no-loss efficiency by increasing $n$ or $m$.
Furthermore, we can see  that the expression for the exact no-loss efficiency of surface($n,m$) becomes much more complicated for $n>3$, because then a square is not necessarily a neighbor to all squares in the previous row anymore. A small summary of these efficiencies can be seen in \autoref{tab:surfaceefficiencyzz1}.
More details on the above interpretation and the calculations can be found in App. \ref{sec:calculations}.

Let us now consider the possibility of combining $ZZ=1$ and $XX=1$ measurements.
When allowing for a fixed combination (still no feedforward) of $ZZ$=1 and $XX$=1 measurements we have many choices where to use which BM.
We have tried all possibilities for small codes and looked at the formations that maximize the no-loss efficiency.
 Then we also tried to understand their efficiency with a combinatorial argument that can be generalized to bigger codes.
 There are many formations that allow for a huge improvement of the no-loss BM efficiency, but the corresponding efficiencies with loss may differ even if the no-loss efficiency is the same for two formations.
 As an example for a formation for which the no-loss efficiency increases we consider measurement formations similar to \autoref{fig:formations}(a) for surface($n,m$), giving a no-loss efficiency of $1-2\cdot4^{-\text{max}(n,m)}$.
 The general idea of these formations is that we perform e.g. $ZZ$=1 measurements along the support of one $\overline{Z}$ so that we obtain the $\overline{ZZ}$-information for sure, we perform $XX$=1 measurements elsewhere, and we choose  $\overline{Z}$ in such a way that only one of the $ZZ$=1 measurements needs to give additional information for also obtaining the $\overline{XX}$-information. Such a formation can also be used for planar color codes as it can be seen in \autoref{fig:formations}(c) for the Steane code.

\begin{theorem}
\label{thm:colorcodexz}
Using an $[N,1,d]$ planar color code one can achieve a logical BM efficiency of $1-2^{-d}$ if $ZZ$=1-BMs are performed on the $d$ qubit pairs belonging to one of the three boundaries of the triangle (see \cref{fig:colorcode488} in App. \ref{a:review}) and by performing $XX$=1-BMs on all other qubit pairs.
\end{theorem}
\begin{proof}
$\overline{ZZ}$ information is available by construction of the BM formation and in order to obtain  $\overline{XX}$ information additional information is only needed from one of the $d$ qubit pairs with $ZZ$=1-BMs.
 In order to calculate the probability that there is at least one qubit pair with additional information, we are using  \autoref{lem:enum}.
 We associate $Z$-codewords with the support of $\overline{X}$ and $X$-stabilizer generators.
 It can be easily seen that only $d-1$ independent $X$-stabilizer generators have support on the boundary and all other $X$-stabilizers are thus irrelevant for the argument.
 $\overline{X}$ has its support on all qubit pairs and hence we have $d$ independent objects with support on the boundary.
 The only possibility to have no $ZZ$=1 result for any of the $d$ qubit pairs is the combination which uses none of the $d$ objects, but there are $2^d$ possibilities and thus the BM efficiency is given by $1-2^{-d}$ .
\end{proof}
\noindent Of course, in \cref{thm:colorcodexz}, $XX$ and $ZZ$ can be exchanged like in the example of \autoref{fig:formations}(c).

There is no experimental reason as to why we should restrict ourselves to using only $XX$=1 and $ZZ$=1 measurements instead of using all possible six BMs that can identify two of the four Bell states by static linear optics.
 In fact, employing these less restricted formations (see \autoref{fig:formations}(b,d,e)), it is possible to obtain even higher no-loss efficiencies for small codes [surface(2,2), Steane, QPC(2,2)] by simply testing all possible formations.
For QPC($n,m$), we were able to generalize this observation.
Note that for QPC it is crucial to include $YY=1$ measurements, which one can easily understand for QPC(2,2) where otherwise (with only $ZZ=1$ and $XX=1$ measurements) the no-loss efficiency remains $\frac{3}{4}$ as for the $ZZ=1$-only case.
\begin{theorem}\label{thm:qpc}
The no-loss BM efficiency of QPC($n,m$) with static linear optics is given by $1-2^{-\left(n+m-1\right)}$ using the measurement formation as given in \autoref{fig:formations}(e).
\end{theorem}
The proof is rather lengthy and it is presented in App. \ref{sec:qpcproof}.
 It is unclear whether there exist even more efficient formations in the no-loss case, because we only found the optimal (transversal) formation  of QPC(2,2) and generalized its formation to QPC($n,m$)  \footnote{As the bounds for the linear-optics BMs derived in Refs. \cite{quant-ph/0007058,jeongfeedforward} exclude feedforward, it is still open whether a logical BM efficiency greater than $1-2^{-(n+m-1)}$ and possibly approaching $1-2^{-nm}$ can be achieved with ancilla-free static linear optics and general QPC($n,m$)\label{footnote:qpc}. For the particular case of QPC(2,2), however, our optimal guaranteed-information formation combined with our reduction of arbitrary transversal linear-optics BMs to the set of guaranteed-information BMs (see the discussion after Theorem 4) implies that $1-2^{-3}=\frac{7}{8}$ is already the ultimate limit for ancilla-free static transversal linear optics.)}.
 Let us now compare our new results for the logical BM efficiencies  obtained for QPC($n,m$) with the previous results.
 Our work, that of Refs. \cite{EwertPRAqpc,Ewertqpcprl}, and Lee et al. \cite{jeongfeedforward} have in common that they all use transversal linear optical BMs mostly without additional ancilla photons.
 Refs. \cite{EwertPRAqpc,Ewertqpcprl} use only standard BMs while here we (similar to Ref. \cite{jeongfeedforward}, but strictly without feedforward) include the possibility of linear optical transformations before standard BMs.
 Thus, importantly, our schemes do not rely upon feedforward operations, whereas especially the scheme of Ref. \cite{jeongfeedforward} does including the no-loss case.
 Using only standard BMs one obtains a no-loss efficiency of $1-2^{-n}$ for QPC($n,m$) \cite{EwertPRAqpc,Ewertqpcprl}.
 If we allow for additional linear optical transformations before the standard BMs, we can achieve an efficiency of $1-2^{-(n+m-1)}$. 
The idea behind this combination of guaranteed-information BMs is that we always obtain $\overline{YY}$-information so that we only need additional $\overline{XX}$ or $\overline{ZZ}$-information, where it is possible that the $ZZ$-measurement results are sufficient for obtaining the $\overline{ZZ}$-information, but the $XX$-measurement results are not sufficient for obtaining the $\overline{XX}$-information. Therefore, we will say that $\overline{ZZ}$ can be measured directly, while $\overline{XX}$ can only be measured indirectly. When considering only $XX$ and $ZZ$ guaranteed-information, it is impossible to obtain the $\overline{YY}$-information without having the logical Bell state identified. 
 Note that this improvement can be achieved solely based on static linear optics. If, however, we further allow for feedforward we can get an efficiency of $1-2^{-nm}$ \cite{jeongfeedforward} \footnote{See previous footnote.} (see \autoref{tab:qpcnolosscomparisonl} for a comparison of logical BM efficiencies of QPC).
 Recall that we only obtained results for arbitrary ($n,m$) in the no-loss case, whereas the works of Refs.  \cite{EwertPRAqpc,Ewertqpcprl,jeongfeedforward} also include loss.
 A counting script enabled us to also calculate the logical BM efficiency including loss for small codes like, for example, QPC(2,2) or QPC(4,2). Although for small losses, we still obtain an improvement with our new static linear-optics methods (see \autoref{fig:qpc22loss} and \autoref{fig:comparisonadvanced}(a)), unfortunately, with increasing loss our generalized scheme's efficiency scales worse (at least for QPC(3,2) and QPC(4,2)) than the scheme using only standard BMs. Next, we discuss the logical BM efficiencies including photon loss in more detail.

\begin{table}[]
\centering
\begin{tabular}{l|l}
QPC(n,m)     &  \(\displaystyle 1-2^{-n} \) \\\hline
surface(2,m) &  \(\displaystyle \frac{1}{2}\left(1+\left(\frac{1}{2}\right)^m\right)\) \\\hline
surface(3,m) &   \(\displaystyle \frac{1}{2}\left(1+\left(\frac{3}{4}\right)^m\right)\) \\\hline
surface(4,m) & \(\displaystyle \approx \frac{1}{2}+\frac{(41+7\sqrt{41})(7+\sqrt{41})^m}{41\cdot 4^{2m+1}}\)
\end{tabular}
\caption{Comparison of the no-loss efficiencies of QPC and planar surface codes when using only ZZ=1 measurements.}
\label{tab:surfaceefficiencyzz1}
\end{table}

\begin{table}[]
\centering
\begin{tabular}{l| l| l| l}
QPC   & static, standard \cite{EwertPRAqpc,Ewertqpcprl} & static, generalized & feedforward  \cite{jeongfeedforward}     \\\hline
(2,2) & \(\displaystyle \frac{3}{4}\)                                              & \(\displaystyle \frac{7}{8}\)                                                                    & \(\displaystyle \frac{15}{16}\)     \\\hline
(3,2) & \(\displaystyle \frac{7}{8}\)                                               & \(\displaystyle \frac{15}{16}\)                                                                    & \(\displaystyle \frac{63}{64}\)      \\\hline
(4,2) & \(\displaystyle \frac{15}{16}\)                                                 & \(\displaystyle \frac{31}{32}\)                                                                   & \(\displaystyle \frac{255}{256}\)    \\\hline
(3,3) & \(\displaystyle \frac{7}{8}\)                                               & \(\displaystyle \frac{31}{32}\)                                                                    & \(\displaystyle \frac{511}{512}\)   
\end{tabular}
\caption{Comparison of logical BM efficiencies in the loss-free
case for some small instances of QPC depending on the use of
(standard or generalized) static or non-static linear optics. The column "static, generalized" refers to the present work.}
\label{tab:qpcnolosscomparisonl}
\end{table}
\renewcommand{\arraystretch}{1.0}

\begin{figure}
\subfloat[\label{sfig:formationsurfacexz}]{%
  \includegraphics[width=0.2\textwidth]{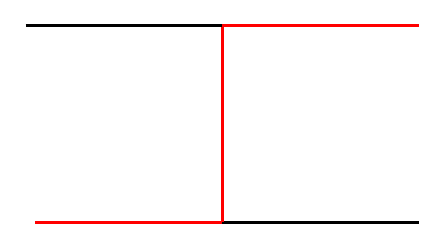}%
}
\subfloat[\label{sfig:formationsurfacexyz}]{%
  \includegraphics[width=0.2\textwidth]{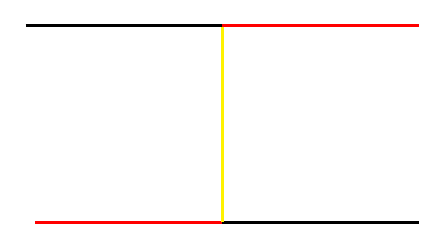}%
}

\subfloat[\label{sfig:formationsSteanexz}]{%
  \includegraphics[width=0.2\textwidth]{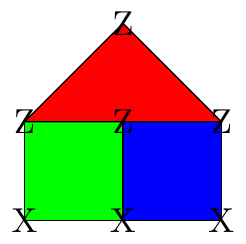}%
}
\subfloat[\label{sfig:formationsSteanexyz}]{%
  \includegraphics[width=0.2\textwidth]{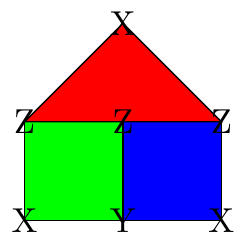}%
}

\subfloat[\label{sfig:formationsqpcnmxyz}]{%
  \includegraphics[width=0.25\textwidth]{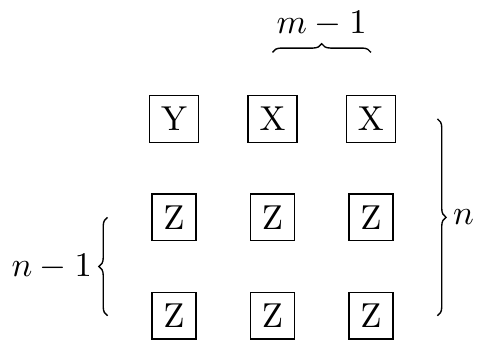}%
}

	\caption{(a) and (b) show good measurement formations for the surface(2,2) code resulting in a no-loss efficiency of $1-2^{-3}$ and $1-2^{-5}$.
		 Black bars denote $XX$=1 measurements, red bars denote $ZZ$=1 measurements, and yellow bars denote $YY$=1 measurements.
		(c) and (d) show good measurement formations for the Steane code resulting in a no-loss efficiency of $1-2^{-3}$ and $1-2^{-5}$.
		 (e) shows the measurement formation which gives a no-loss efficiency of $1-2^{-(n+m-1)}$ for QPC($n,m$).}
	\label{fig:formations}	
\end{figure}

\begin{figure}
\includegraphics[width=0.4\textwidth]{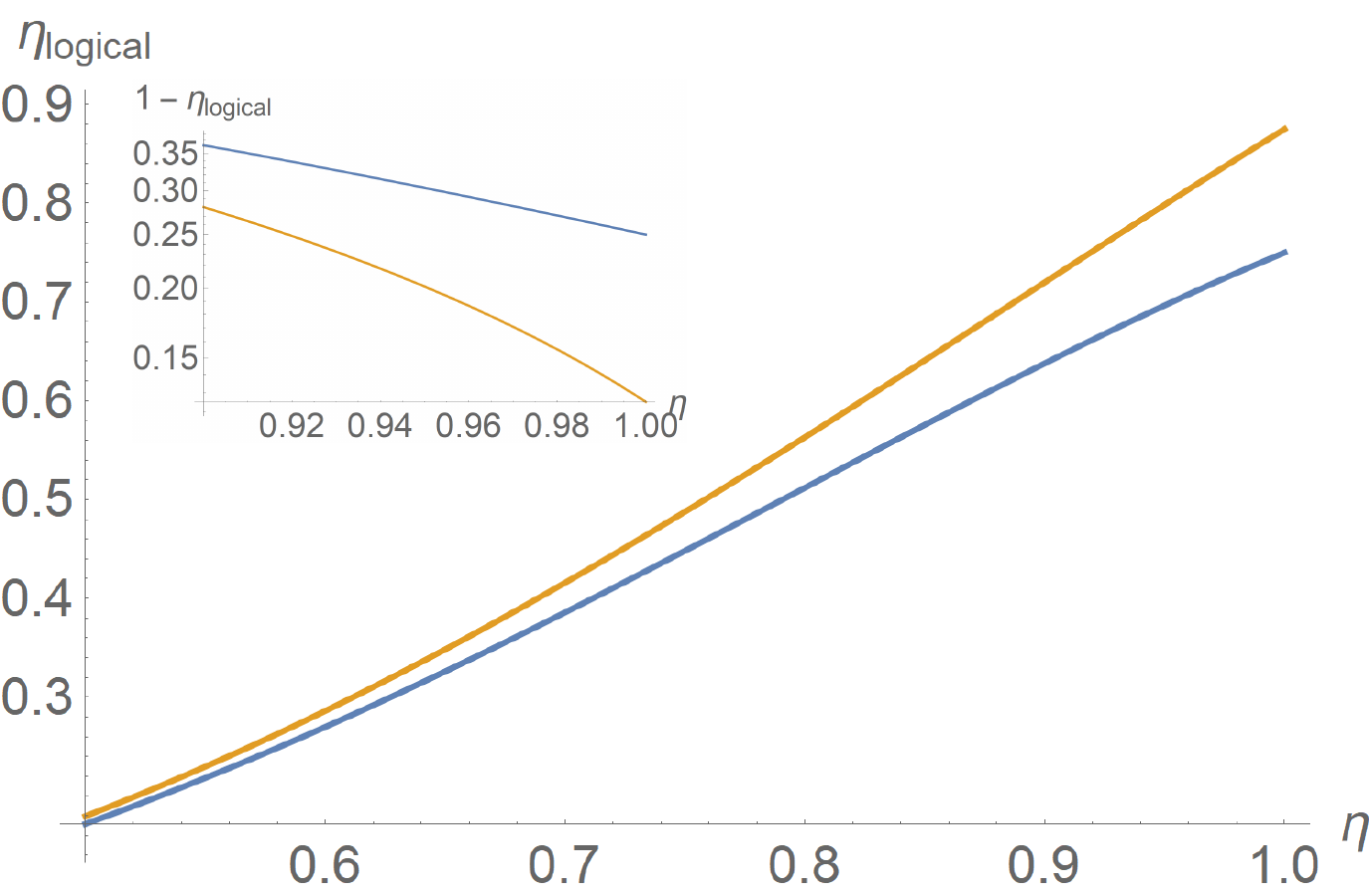}
\caption{Efficiency of the logical BM for QPC(2,2) in the presence of photon loss.
The orange line represents the use of our new static linear-optics logical BM employing a combination of $XX=1$, $YY=1$, and $ZZ=1$ BMs, while the blue line corresponds to the already known logical BM efficiency where only $ZZ=1$ BMs are used with static linear optics \cite{Ewertqpcprl}. }
\label{fig:qpc22loss}
\end{figure}

\subsubsection{Logical BM efficiencies including photon loss}
We shall now turn to the more general situation including photon loss, which is most important for quantum communication.
In principle, it is also easy to include photon loss in the logical BM efficiency.
We only need to go through all possible patterns of erasures and codewords, and check if we obtain $\overline{XX}$ and $\overline{ZZ}$-information taking into account that we do not get any information from erased qubit pairs. 
However, the photon loss may destroy the structure of the code that we have exploited so far for calculating the logical linear-optics BM efficiency efficiently in the loss-free case.
Therefore, we are only able to perform these calculations via python counting scripts, and so we will compare small codes like Steane, QPC(4,2) and surface(3,2) that are based on a similar amount of physical qubits.
It can be seen in \autoref{fig:comparisonadvanced} ($\eta$ denotes the transmission within a single repeater segment of length $L_0$)  that these codes, when employed in a 3\textsuperscript{rd} generation quantum repeater, cannot outperform direct, unencoded DR-qubit transmission considering only linear optical BMs without ancilla photons and without feedforward.
Nonetheless, note that the guaranteed-information formation as described in \autoref{thm:qpc} outperforms the previous results of Refs. \cite{EwertPRAqpc,Ewertqpcprl} for QPC(4,2) for $\eta>0.95$ (assuming only loss and no other types of errors).
As one expects it is possible to outperform direct transmission by making use of advanced BMs with sufficiently high $p_\text{adv}$, because we could use many ancilla photons so that we obtain almost ideal measurements. 
Interestingly, however, only $p_\text{adv}=0.5$ is needed so that the Steane code (similar to QPC(4,2)) outperforms direct transmission and such an advanced BM is still experimentally feasible, since it only requires static linear optics and additional single photon sources \cite{EwertAdvancedBM}.

\autoref{tab:numberofrepeater_Steanel} shows the minimal number of stations of a linear optical 3\textsuperscript{rd} generation quantum repeater based on the Steane code for which it becomes more efficient to use the repeater than using all resources for parallel direct transmission instead.
 Note that we employed the Grice scheme \cite{Grice} for obtaining $p_{\text{adv}}>0$ and we only considered the number of ancilla photons as a resource cost, ignoring the actual difficulty of producing an entangled auxiliary state.
 Therefore, we would obtain a different number of required repeater stations if we consider the scheme of Ref. \cite{EwertAdvancedBM} instead, where the preparation of the ancilla state is easier (at least for $v$=1), but more photons are needed.
The repeaterless bound \cite{PLOB} can be exceeded (assuming a secret key fraction of unity) with 123 stations and ideal static linear optical BMs involving only auxiliary single photon sources.

\begin{figure}[h]
\subfloat[\label{sfig:comparisonnoadvanced}]{%
  \includegraphics[width=0.4\textwidth]{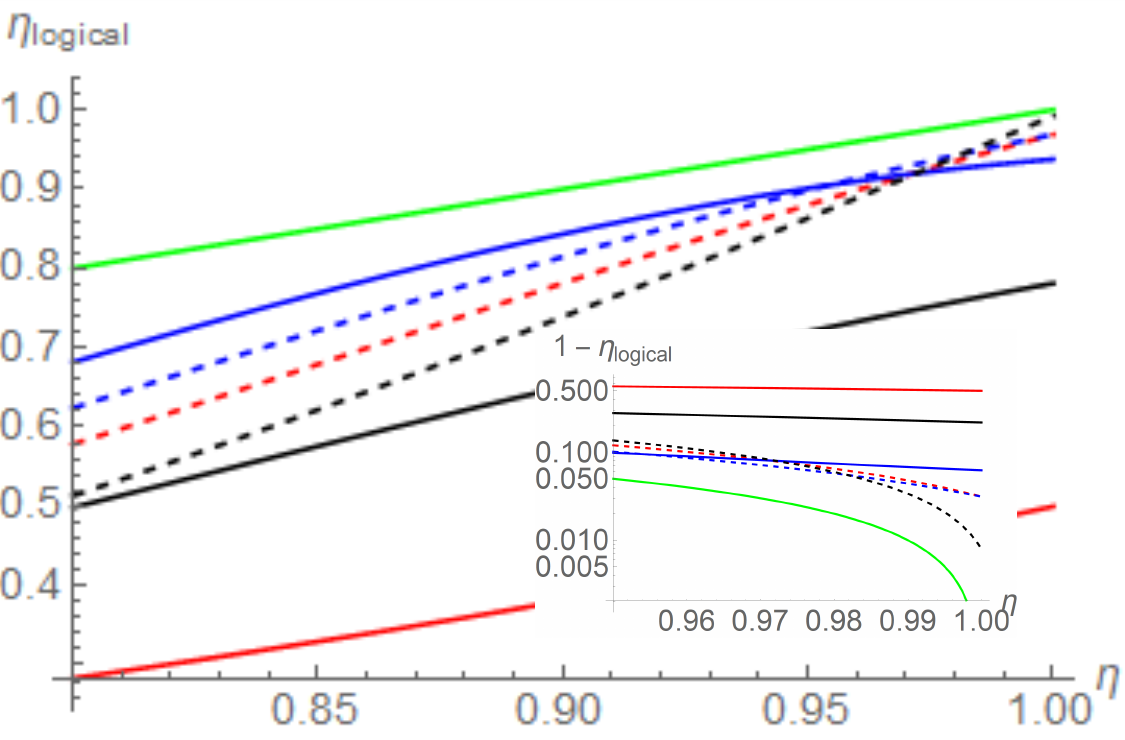}%
}

\subfloat[\label{sfig:comparison50advanced}]{%
  \includegraphics[width=0.4\textwidth]{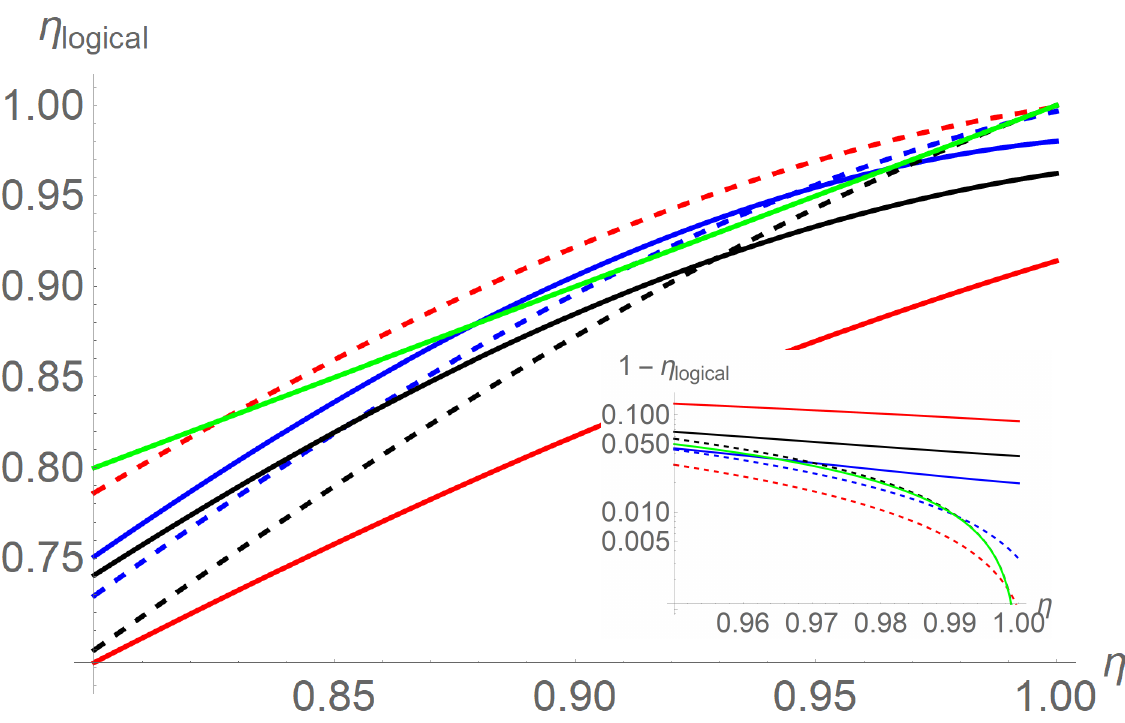}%
}

\subfloat[\label{sfig:comparison75advanced}]{%
  \includegraphics[width=0.4\textwidth]{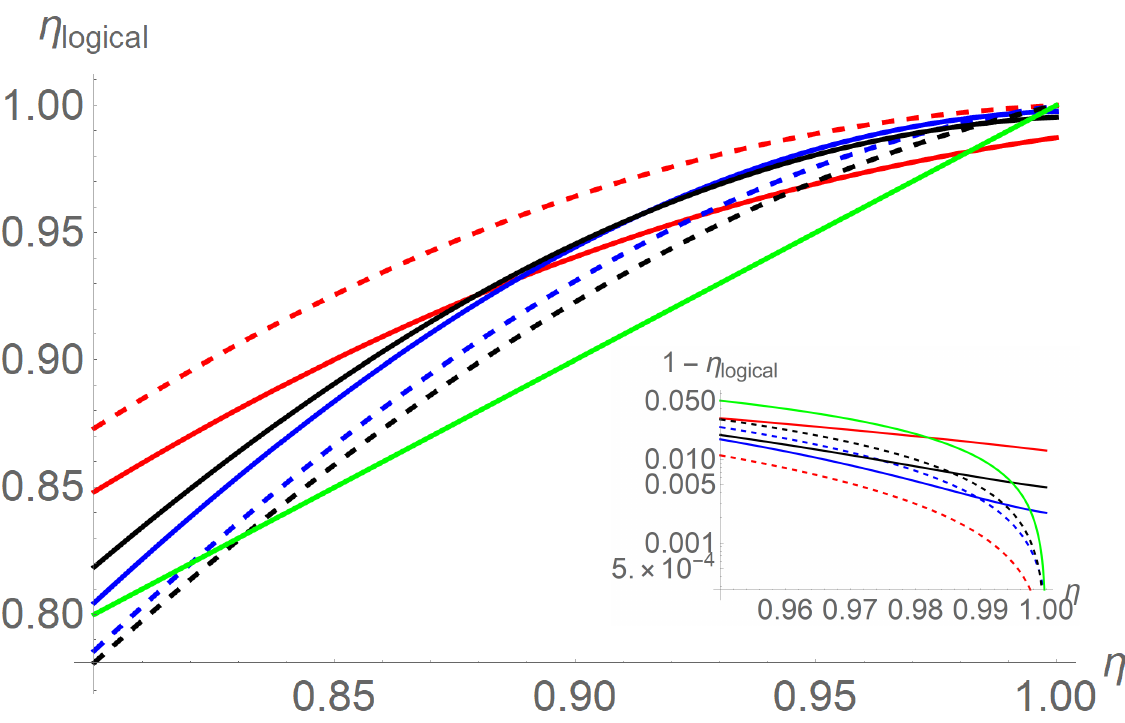}%
}

	\caption{Efficiency of the logical BM in the presence of photon loss.
		Solid  lines denote the use of only $ZZ$=1 measurements and dashed lines denote the use of $XX$,$YY$ and $ZZ$=1 measurements.
		 The green line $\eta_\text{log}=\eta$ corresponds to direct, unencoded DR-qubit transmission without a quantum repeater and is therefore the benchmark that should be outperformed.
		 Red corresponds to the Steane code, blue corresponds to QPC(4,2) and black corresponds to surface(3,2). (a) $p_{\text{adv}}=0$, (b) $p_{\text{adv}}=0.5$, (c)$p_{\text{adv}}=0.75$.}
	\label{fig:comparisonadvanced}	
\end{figure}

\begin{table}
\begin{tabular}{|c|c|c|c|c|}
	\hline 
	$v$ & best repeater & best gain & needed repeater& needed repeater  \\ 
	& spacing[km]	&	& stations (cost)&stations (PLOB) \\
	\hline 
	1 & 1.99 & 1.02489 & 377&123 \\ 
	\hline 
	2 & 4.46 & 1.09239 & 98&34  \\ 
	\hline 
	3 & 5.89 & 1.14495 & 66&23  \\ 
	\hline 
	4 & 6.60 & 1.17687 & 59&19  \\ 
	\hline 
	5 & 6.96 & 1.19432 & 58&17  \\ 
	\hline 
	6 & 7.13 & 1.20341 & 59&17  \\ 
	\hline 
	7 & 7.22 & 1.20805 & 62&16  \\ 
	\hline 
\end{tabular} 
\caption{Comparison of the performance of the Steane code with various parameters of $p_\text{adv}=1-2^{-v}$ of the Grice scheme using the optimal formation of guaranteed-information measurements. 
	The repeater spacing is calculated numerically, so that the ratio between code-based transmission to direct transmission is maximized.
	Using this maximum gain per repeater station,  the number of repeater stations is given that is needed to outperform parallel direct transmission taking into account the cost of the extra ancilla photons. 
	Adding more ancilla states reduces the number of needed repeater stations up to $v$=5, because the improvement to direct transmission increases. 
	However, adding more ancillae increases the number of required stations, because the ratio of improvement is bounded by the perfect-measurement case and more resources are needed to achieve a higher $v$. It is also shown how many repeater stations are needed to beat the PLOB-bound \cite{PLOB} under ideal conditions. }

\label{tab:numberofrepeater_Steanel}
\end{table}
\section{Conclusion}

Compared to previous works on logical Bell measurements with static linear optics \cite{Ewertqpcprl,EwertPRAqpc}, we have presented an alternative approach for identifying logical Bell states by looking at the codes in the stabilizer formalism instead of the state picture.
This enabled us to generalize the previous schemes for the specific code family QPC to arbitrary CSS codes. 
It was possible then to investigate the influence of different linear optical BMs on the logical BM efficiency.
For general CSS codes, we proposed to calculate the codewords of the two classical codes $C_X$ and $C_Z$ instead of finding a decomposition of a logical Bell state in the Bell basis as it was done in Refs. \cite{Ewertqpcprl,EwertPRAqpc}.
 We showed that  different codes need different choices of linear optical BMs in order to achieve a good logical BM efficiency.
 We also demonstrated that the possibility to increase the logical BM efficiency arbitrarily close to unity only based on linear optical BMs without ancilla photons and without feedforward is not a unique feature of QPC, although that code does have the best scaling of the no-loss efficiency (in comparison to the other codes we considered).

More specifically, we showed that through our generalizations the logical BM efficiency for any QPC can be significantly enhanced in the no-loss case and noticeably improved for small instances of QPC in the low-loss case. The logical BM efficiency of two-dimensional planar color codes approaches unity with our method, as opposed to the $\frac{1}{2}$-limit that exists for the standard method, as we proved. Several small codes such as QPC(4,2), surface(3,2), and the prominent Steane code (the smallest instance of a color code) were shown to exceed the repeaterless bounds in an encoded one-way all-optical quantum repeater, provided that the physical BMs involved are assisted by ancilla photons and selected according to our new method.

A drawback of our method is that it is a rather brute-force approach and codes typically have a lot more algebraic structure. We only utilized the feature of the codes being  CSS codes.
 Therefore, it might still be possible to apply our method and adapt it to particular subclasses of CSS codes resulting in a much better computational performance. 

Furthermore, we showed that the whole set of the most general ancilla-, feedforward-free transversal linear optical BMs can be reduced to the set of guaranteed-information BMs, as introduced in our work, when looking for optimal logical BM efficiencies. In addition, it can also be useful to consider BMs whose efficiency is independent of the input state. This allowed us to reduce the linear-optics constraint to an erasure channel of transmittance $p_{BM}$ in addition to the actual photon loss. Therefore, this simplification only gives interesting results when ancillae are employed. 

Finally, it is conceivable that one obtains a higher logical BM efficiency if one does not rely on transversal BMs, but a more complex linear optical network.

\begin{acknowledgments}
We thank F. Ewert for helpful discussions. We also thank  Q.com and Q.Link.X from the BMBF in Germany for funding.
\end{acknowledgments}

\appendix
\section{\label{a:review}Brief review of relevant codes}

\subsection{QPC}

The quantum parity code (QPC) was introduced in Ref. \cite{qpcintro}, but we use a modified version like that in Ref. \cite{PhysRevLett.112.250501}.
The QPC($n,m$) can be seen as a concatenation of two repetition codes of length $m$ and $n$ with respect to the $Z$- and $X$-basis. 
 Thus, the code has a block structure.
 The first level of encoding is the physical qubit (e.g. dual-rail encoding). The next level is the block level where a block $\{\ket{0}^{(m)}:=\ket{0}^{\otimes m},\ket{1}^{(m)}:=\ket{1}^{\otimes m}\}$ is given by $m$ repetitions in the $Z$-basis. The third and last level is the logical level\newline $\{\frac{1}{\sqrt{2}}\left(\ket{0}^{(n,m)}\pm \ket{1}^{(n,m)}\right):=\frac{1}{\sqrt{2^n}}\left(\ket{0}^{(m)}\pm \ket{1}^{(m)}\right)^{\otimes n} \}$ and it is given by $n$ repetitions of blocks in the $X$-basis. The first repetition code is used to correct $X$-errors while the second repetition code is used to correct $Z$-errors, and thus QPC($n,m$) is an $[n\cdot m,1 ,\text{min}(n,m)]$-code.
 A very famous special case of this family of codes is QPC(3,3), also known as the  Shor code \footnote{Notice that a code is just defined as a subspace and relabeling codewords does not matter at all. The Shor code in Ref. \cite{Nielsen:2011:QCQ:1972505} is equal to our defined QPC(3,3) up to a logical Hadamard gate.}. \newline
  It is also easy to define the code in the stabilizer formalism. Each block of $m$ qubits is stabilized by the $m-1$ stabilizer generators $ Z_{l,j}Z_{l,j+1} \quad \forall j \in \{1,\dots,m-1\}$. We can define Pauli operators for each block, for example,  as
\begin{equation}
\begin{aligned}
X^{(m)}_l&:=\prod_{j=1}^m X_{l,j}\\
Z^{(m)}_l&:=Z_{l,1}
\end{aligned}
\label{eq:qpc_block}
\end{equation}
where $l$ denotes the different blocks.
Due to the concatenation of the codes, we can obtain the $X$-stabilizer generators as 
$ X^{(m)}_j X^{(m)}_{j+1} ,\quad \forall j \in \{1,\dots,n-1\}$.
As a consequence, we have $n\cdot (m-1)+n-1=nm-1$ stabilizer generators and we use $n\cdot m$ physical qubits.
Thus, we can see again that this code encodes one logical qubit, $k=1$.
Similar to \cref{eq:qpc_block}, we can see that logical Pauli operators of QPC are given, for example, as \footnote{Note that the product of a logical operator and a stabilizer is also a logical operator.}
\begin{equation}
	\begin{aligned}
	\overline{X}&=X^{(m)}_1\\
	\overline{Z}&=\prod_{j=1}^{n} Z^{(m)}_j
	\end{aligned}
\end{equation}

\subsection{Planar surface code}
Surface codes are defined via a cellulation of a surface.
We define a cellulation as a division of the surface into polygonal cells such that all cells meet edge-to-edge and vertex-to-vertex.
Every edge corresponds to a physical qubit while faces or vertices can be associated with $Z$- or $X$-stabilizers, respectively. When considering a surface code on a torus one obtains the toric code. In our work, we will consider a  planar two-dimensional surface with boundaries. 
However, there are different kinds of quantum codes which are called planar surface codes. On the one hand, there are codes that use surfaces with always the same type of boundaries including holes.
 On the other hand, there are codes that use surfaces with different types of boundaries having no holes. In this paper, we are referring to the latter and we assume a square tessellation.
The dual lattice is obtained by mapping faces to dual vertices, edges to dual edges, and vertices to dual faces. Dual vertices are connected by dual edges if their corresponding faces on the primal lattice are adjacent.
 There are two kinds of different boundaries. First there are smooth boundaries, which are named this way, because they appear smooth on the primal lattice, but appear rough on the dual lattice. There are also rough boundaries named for analogous reasons. An example of such boundaries and a combination of both boundaries can be seen in \autoref{fig:planarboundaries}.
 \begin{figure}
 	\centering
 	\includegraphics[width=0.8\linewidth]{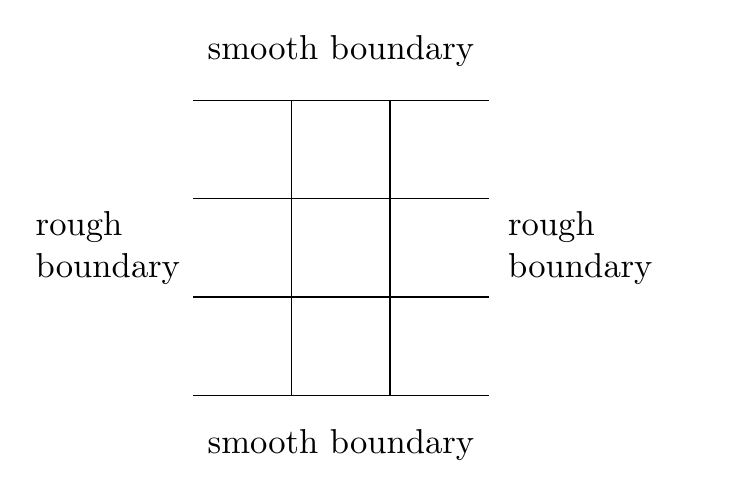}
 	\caption{Lattice on a planar surface with two smooth and two rough boundaries. Smooth boundaries are rough in the dual lattice and rough boundaries are smooth in the dual lattice.}
 	\label{fig:planarboundaries}
 \end{figure}
The stabilizers of this code are generated in the following way. The $Z$-stabilizer generators are given via faces on the primal lattice, while the $X$-stabilizer generators are given via vertices.
As an example let us construct the planar surface code ($n$=3, $m$=2).
 \begin{figure}
 	\centering
 	\includegraphics[width=0.3\linewidth]{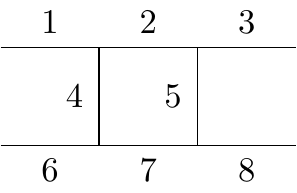}
 	\caption{Numbering of the surface(3,2) code.}
 	\label{fig:surfacenumbering}
 \end{figure}
The $Z$-stabilizer generators (the numbering of the qubits is defined in \autoref{fig:surfacenumbering}) are given by $Z_1Z_4Z_6,Z_2Z_4Z_5Z_7,Z_3Z_5Z_8$ and the $X$-stabilizer generators are given by $X_1X_2X_4,X_2X_3X_5,X_4X_6X_7,X_5X_7X_8$.
Thus, we have $N-k=7$ stabilizer generators defining surface(3,2) that encodes $k=1$ logical qubits into $N$=8 physical qubits. 
$\overline{Z}$ is given by e.g. $Z_1Z_2Z_3$ (up to the multiplication of a a stabilizer) and $\overline{X}$ is given by e.g. $X_1X_6$. This means the minimum weight of $\overline{Z}$ in a planar surface($n,m$) is given by $n$ while $m$ gives the minimum weight of $\overline{X}$. It is a nice feature of surface codes that logical $\overline{Z}$ and $\overline{X}$ can be associated with elements of the first homology group of the primal/dual lattice of the surface. Nice introductions to these codes (and also to color codes, which are closely related) can be found in Refs.  \cite{topocodelecture,lidar2013quantum,1410.0069}.

\subsection{Planar color code}

Color codes are also defined via cellulations of surfaces similar to surface codes. However, we need to assume a three-valent cellulation and it needs to be three-colorable. Then we can identify each vertex with a physical qubit and every face corresponds to an $X$- and a $Z$-stabilizer. This means all qubits that correspond to vertices spanning a face are in the support of this face stabilizer. An exemplary graph of planar color codes is shown in \autoref{fig:colorcode488}. An $\overline{X}$- or a $\overline{Z}$-operator is given by applying $X$ or $Z$ on all qubits on the code. It can also be seen in \autoref{fig:colorcode_steane} that the smallest color code with $d=3$ is the well-known Steane code. According to the construction of the planar color codes, the $X$-stabilizer generators are given by $\{X_1X_3X_5X_7,X_2X_3X_6X_7,X_4X_5X_6X_7\}$ in agreement with the Steane code $X$-stabilizer generators given in \cite[p. 456, Fig. 10.6]{Nielsen:2011:QCQ:1972505}. An example of a low weight $\overline{X}$-operator would be $X_1X_2X_3$ confirming $d=3$. Similarly, the $Z$-operators of this planar color code are consistent with  those of the Steane code given in Ref. \cite{Nielsen:2011:QCQ:1972505}.
\begin{figure}
	\centering
	\includegraphics[width=0.8\linewidth]{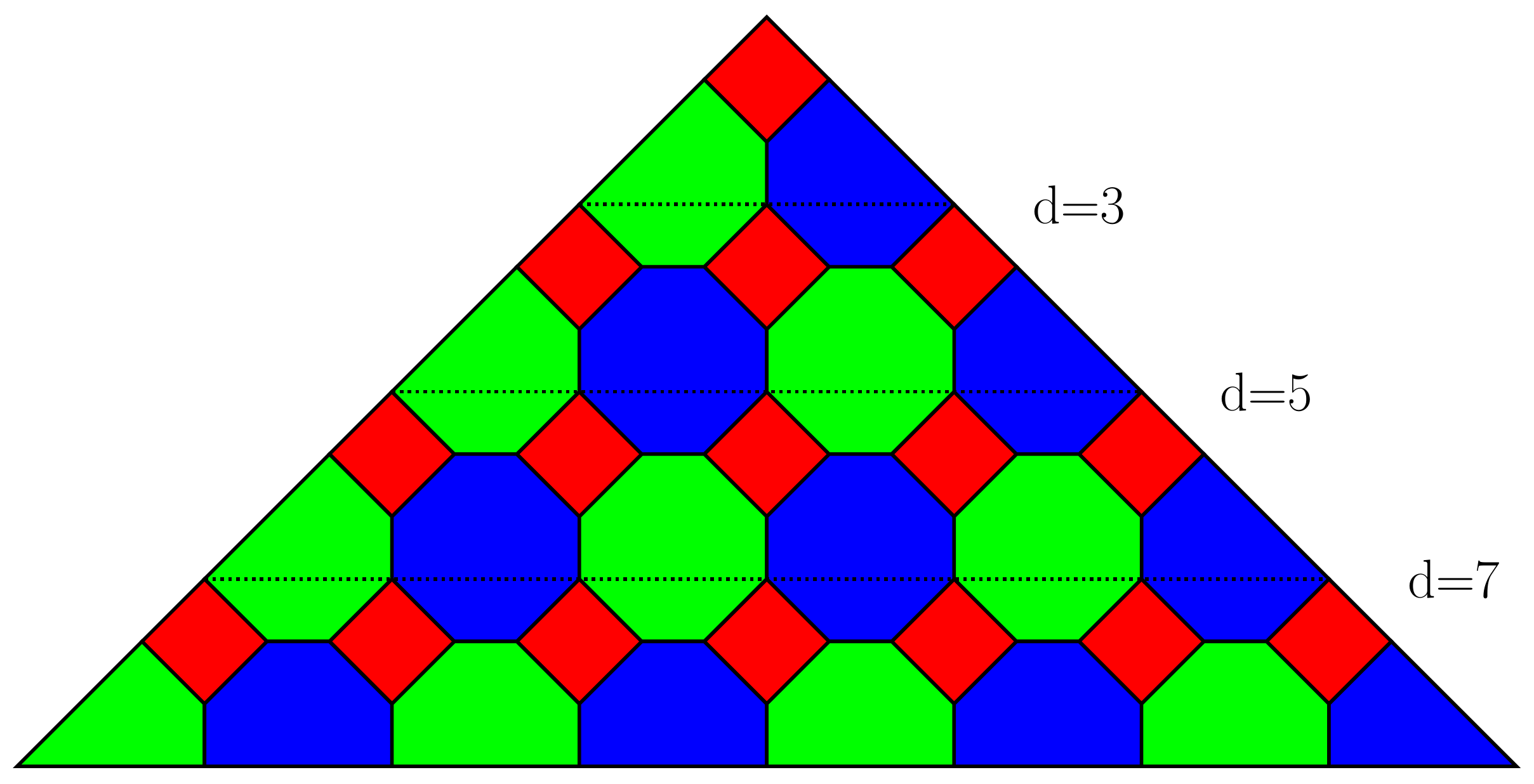}
	\caption{Family of 4-8-8 tessellated triangular color codes. E.g. a code with a code distance of 5 can be obtained by only considering the triangle above the dotted line labeled by ’d=5’.}
	\label{fig:colorcode488}
\end{figure}

\begin{figure}
	\centering
	\includegraphics[width=0.8\linewidth]{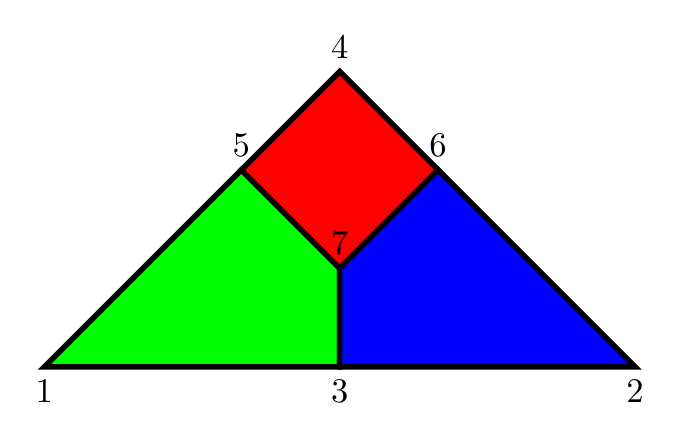}
	\caption{The Steane code as a planar color code. The qubit labeling was chosen in this particular way in order to be consistent with the definition of the Steane code given in Ref. \cite{Nielsen:2011:QCQ:1972505}.}
	\label{fig:colorcode_steane}
\end{figure}

\section{Unconstrained BMs}
Here we will give a short review of the logical BM efficiencies in the absence of the linear-optics constraint for the same small codes that have been compared with each other including the constraint in the main text in \autoref{fig:comparisonadvanced} (a).
In \autoref{fig:comparisonnonlinear} it can be seen that the Steane code achieves always the best transmission while QPC(4,2) achieves always the worst.
However, comparing this ranking with the efficiencies when considering only $ZZ$=1 measurements, one can see that the order of the ranking is just reversed  (as already given by the no-loss efficiencies).
Interestingly, when allowing for combinations of different guaranteed-information BMs the Steane and surface(3,2) codes are at least as good as QPC(4,2) in the loss-free case, but for $\eta<0.98$ QPC(4,2) gives the best performance.
Therefore, it is not possible to estimate a code's performance assuming linear optical BMs and photon loss by knowing only, on the one hand, the  performance including loss without the linear-optics constraint and, on the other hand, the loss-free logical BM efficiency with the linear-optics constraint.
 Ref. \cite{eppingcomparison} compared the costs of quantum repeaters for different quantum codes without linear-optics constraint and  concluded that the 23-qubit Golay code seems to be a really good code for large overall distances (see also \autoref{fig:comparisonnonlinear}). However, \autoref{col:xzgeo} applies to this code, making the code useless for a quantum repeater when considering the same guaranteed-information BMs on all qubit pairs.
As the 23-qubit Golay code is rather large, we did not search for a scheme with different guaranteed-information BMs, neither without nor with loss.

\begin{figure}
\centering
\includegraphics[width=0.8\linewidth]{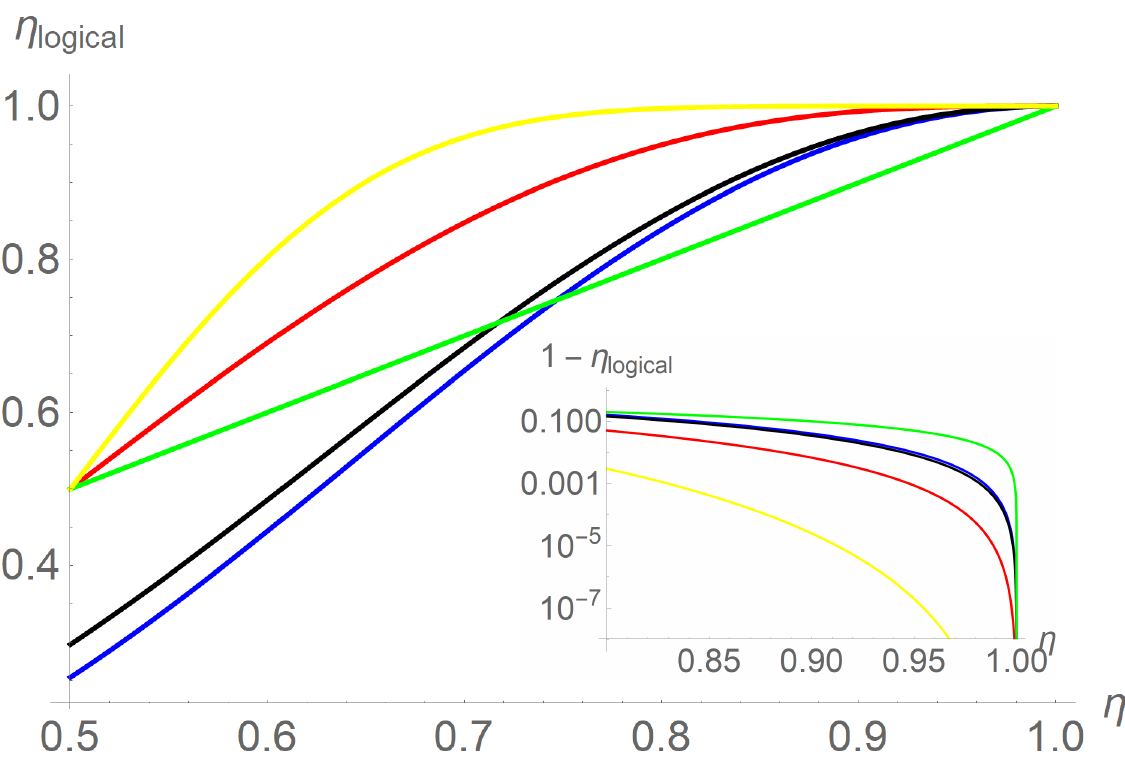}
\caption{Performance of the Steane code (red), surface(3,2) (black), 23-qubit Golay code (yellow), and QPC(4,2) (blue) with unconstrained BMs. Green denotes direct transmission without a code. We also showed the Golay code for a comparison with larger codes.}
\label{fig:comparisonnonlinear}

\end{figure}

\section{Calculations for planar surface code}
\label{sec:calculations}

The planar surface code is a CSS code encoding one logical qubit. When using only $ZZ$=1 BMs, we always get information about $\overline{ZZ}$ and only have to see when we also get information about $\overline{XX}$.
Thus, we can use \autoref{thm:count} and \autoref{lem:enum} to obtain an efficiency of at least $\frac{1}{2}$ and we only have to count the possibilities corresponding to products of $X$-stabilizer generators, such that these fully cover an $\overline{X}$. 
The $X$-stabilizers of the surface code are faces in the dual lattice and $\overline{X}$-operators are strings that connect the two boundaries. 
Thus, one can see that the conditions for a successful logical BM are the same as if one has a board with $(n-1)\cdot m$ squares (each corresponding to an $X$-stabilizer on the dual lattice) where one wants to cross the board (along the side with distance $m$) and one can only walk on squares which are 1 and to adjacent squares which are also 1 (moving diagonally is allowed). Examples which cases are decodable or undecodable can be seen in \autoref{fig:toric_lin_comb_overall} .\newline

\begin{figure}
\subfloat[\label{sfig:toric_lin_comb_allowed_1}]{%
  \includegraphics[width=0.3\linewidth]{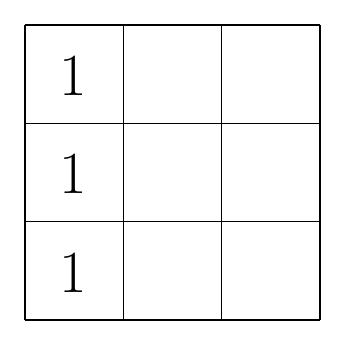}%
}
\subfloat[\label{sfig:toric_lin_comb_allowed_2}]{%
  \includegraphics[width=0.3\linewidth]{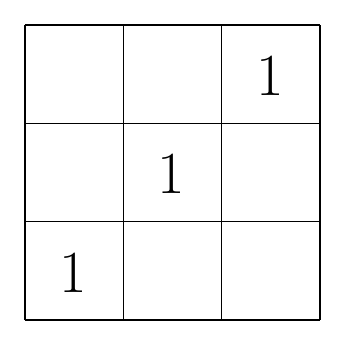}%
}

\subfloat[\label{sfig:toric_lin_comb_allowed_3}]{%
  \includegraphics[width=0.3\linewidth]{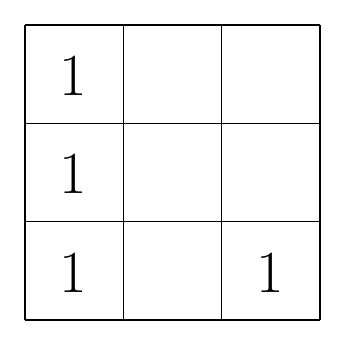}%
}
\subfloat[\label{sfig:toric_lin_comb_allowed_4}]{%
  \includegraphics[width=0.3\linewidth]{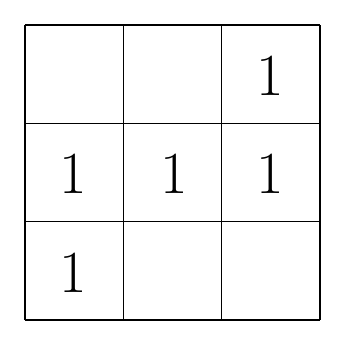}%
}

\subfloat[\label{sfig:toric_lin_comb_forbidden_1}]{%
  \includegraphics[width=0.3\linewidth]{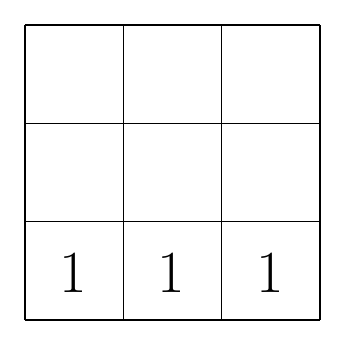}%
}

	\caption{Several example configurations for the combinatorics of the logical BM efficiency for the planar surface code with linear optics. 
	All empty squares have the value 0. In general, there are ($n$-1) columns and $m$ rows. The examples (a)-(d) are allowed, whereas combination (e) is not allowed, i.e. it corresponds to a failed logical BM.}
	\label{fig:toric_lin_comb_overall}	
\end{figure}
 The remaining problem is only of a combinatorial kind getting more and more complicated for increasing code sizes. Therefore, here we discuss only for some special cases how to calculate the efficiency. The surface($n,1$)-code is the same code as the QPC($n,1$)-code and thus the efficiency is known for this case. Now we are looking at cases of arbitrary $m$ and fixed values of $n$ and assume $p_\text{adv}=0$.

\begin{calculation}
The logical BM efficiency for the planar-square-surface(2,m) code using only $ZZ$=1-BMs is given by 
\begin{equation}
\frac{1}{2}\left(1+\left(\frac{1}{2}\right)^m\right)\,.
\end{equation}
\end{calculation}
\begin{proof}
The summand of 1 comes from the fact that every codeword which can be associated with $\overline{X}$ results in a successful logical BM and the summand of $\frac{1}{2^m}$ comes from the fact that we have a board of 1$\cdot m$ squares where all squares have to be 1 in order to be able to connect both sides and there are $2^m$  possibilities.
\end{proof}

\begin{calculation}
The logical BM efficiency for the planar-square-surface(3,$m$) code using only $ZZ$=1-BMs is given by 
\begin{equation}
\frac{1}{2}\left(1+\left(\frac{3}{4}\right)^m\right)\,.
\end{equation}
\end{calculation}
\begin{proof}
The summand 1 comes again from codewords that correspond to $\overline{X}$ and the summand $\left(\frac{3}{4}\right)^m$ comes from the fact that we have $m$ rows of 2 squares per row. The only forbidden combination in each row then is 00 and the proportion  of not having 00 in any row is  $\left(\frac{3}{4}\right)^m$.
\end{proof}

\begin{calculation}\label{thm:toric4m}
The logical BM efficiency for the planar-square-surface(4,m) code using only $ZZ$=1-BMs is given by 
\begin{widetext}
\begin{equation}
\frac{1}{2}+\frac{1}{41\cdot 4^{2m+1}}\left((41-7\sqrt{41})(7-\sqrt{41})^m+(41+7\sqrt{41})(7+\sqrt{41})^m\right)\,.
\end{equation}
\end{widetext}
\begin{proof}
	
Observe that the possibility of only 0s in a row is not allowed and with $n$=4 it is not the case anymore that a square is adjacent to all squares in a previous row. 
Thus, we are splitting the combinations in a row into two classes `z' and `b'. 
Elements of class `z' are not connected to all squares in the previous row, i.e. the two combinations 100 and 001 belong to class `z'. 
Elements of class `b' are connected to all squares in the previous row and there are 5 elements in this class.
Let $z_j$ be the number of possible configurations that end with a fixed element of class `z' (e.g. 100) in the j\textsuperscript{th} row and let $b_j$ be the number of possible configurations that end with a fixed element of class `b' (e.g. 111) in the j\textsuperscript{th} row.\\
In order to calculate $z_j$ and $b_j$ we are setting up a recursive system of equations with the starting condition $z_1=b_1=1$:
\begin{eqnarray}
\label{eq:toric4m}
\begin{pmatrix} z_j \\ b_j \end{pmatrix}&=&\begin{pmatrix} 1 & 5 \\2 & 5\end{pmatrix}   \begin{pmatrix} z_{j-1} \\ b_{j-1} \end{pmatrix}-  \sum_{k=1}^{j-2}  \begin{pmatrix} z_k \\0 \end{pmatrix}\nonumber \\  &\stackrel{j>2}{=}&\begin{pmatrix} 1 & 5 \\2 & 5\end{pmatrix}   \begin{pmatrix} z_{j-1} \\ b_{j-1} \end{pmatrix}-\begin{pmatrix} b_{j-2} \\0 \end{pmatrix}\,.
\end{eqnarray}
The part with the matrix counts the possibilities to connect a row with the previous one, such that there exists at least one square in each row and these two squares are neighbors, and also considering the class of the combination in the previous row. 
This condition is necessary for being able to cross the whole board, but it is not sufficient.
Imagine a piece of 3 rows. Row 1 and 2 are connected on the left side while row 2 and 3 are connected on the right side, and left and right sides are not connected in row 2. 
Thus, the path crossing the whole field is interrupted and all actual correct combinations before row 2 do not help crossing the field (see \autoref{fig:toric_fail_subtraction}).
This gives rise for the subtraction of $z_1$ and one has to sum over all possible places of such an interruption.
 One can also simplify the recursion formula by using the identity \begin{eqnarray}
b_j &=&2z_{j-1}+5b_{j-1} \nonumber \\&=&\underbrace{\left(z_{j-1}+5b_{j-1}-\sum_{k=1}^{j-2}z_k \right)}_{\scalebox{1}{$z_j$}}+z_{j-1}+\sum_{k=1}^{j-2}z_k\nonumber \\ &=&\sum_{k=1}^{j}z_k\,,
\end{eqnarray}which can be obtained easily using the recursion formula.
This simplified  recursion formula was solved by Mathematica and the solution of $z_m$ and $b_m$ was obtained.
The number of all allowed combinations is $2z_m+5b_m$, because there are two elements in class `z'  and five elements in class `b'. Dividing by the number of all possibilities and adding $\frac{1}{2}$ one gets the stated result.

\end{proof}
\end{calculation}
\begin{figure}
	\centering
	\subfloat{
		\includegraphics[width=0.45\linewidth]{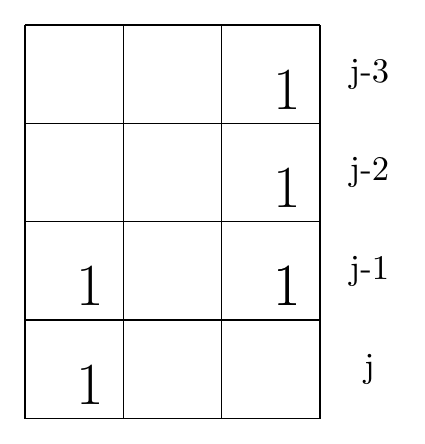}
		\label{fig:toric_4_fail}
		}
	\subfloat{
		\includegraphics[width=0.45\linewidth]{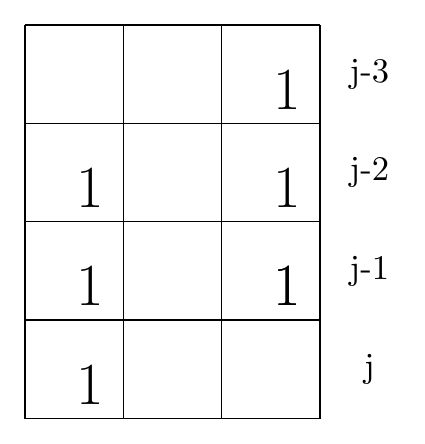}
		\label{fig:toric_4_fail2}
		}
	\caption{Two configurations which give rise to the subtraction in \cref{eq:toric4m}.
		(a) The interruption takes place in the (j-1)\textsuperscript{th} row such that a number of $z_{j-2}$ counted combinations is invalid.
		(b)The interruption takes place in the  (j-2)\textsuperscript{th} row such that a number of $z_{j-3}$ counted combinations is invalid.
		Taking care of all these cases gives the sum in \autoref{eq:toric4m}. }
	\label{fig:toric_fail_subtraction}	
\end{figure}

\begin{calculation}\label{thm:surfacenmefficient}
Using  $ZZ$=1-BMs along the diagonal of the lattice and $XX$=1-BMs elsewhere the efficiency of the planar surface($n,n$)-code using linear optics with $p_{\text{adv}}=0$ and without loss is given by
\begin{equation}
1-2 \cdot 4^{-n}\,.
\end{equation}
\end{calculation}
\begin{proof}
We choose to perform  $ZZ$=1 measurements along the diagonal \footnote{We can also exchange $ZZ$ and $XX$ measurements and the argument remains the same if we look at the dual lattice. } and $XX$=1 measurements elsewhere.
Therefore we know the measurement outcome of $\overline{ZZ}$ and we need only one $ZZ$=1 measurement to give full information in order to obtain $\overline{XX}$ information.
In order to count the fraction of successful identifications of the logical Bell state, we make again use of \autoref{lem:enum}.
 The diagonal is also the support of a $\overline{Z}$-operator and $\overline{Z}$ and $\overline{X}$ anticommute.
 This means that those operators have at least one common physical qubit in their support.
 As a consequence every $\overline{X}$ crosses the diagonal at least once and $Z$-codewords that correspond to such a crossing give us full information such that we obtain the $\overline{XX}$ information.
 This argument is similar to \autoref{thm:count} but now for a combination of $ZZ$=1 and $XX$=1 measurements.  We now have to consider codewords that correspond to products of $X$-stabilizer generators.
  Stabilizer generators that have no qubits of the diagonal in their support do not have an effect for obtaining full information of a BM on a qubit pair on the diagonal and can be ignored. Codewords that correspond to a single $X$-stabilizer generator along the diagonal correspond to obtaining $\overline{XX}$ information.
  Thus, we will check if products of multiple $X$-stabilizer generators along the diagonal can generate an operator whose support does not contain the diagonal.

  We will start with a single $X$-stabilizer generator whose support contains the diagonal and we try to remove  its support on the diagonal by multiplying additional $X$-stabilizer generators.
  Notice that we need to multiply neighboring $X$-stabilizer generators in order to clear the diagonal, but this multiplication adds two new edges of the diagonal to the support of the product operator.
  We could repeat this step again and again in order to push these two edges further away, but since this surface code is defined on a plane with boundaries, we reach this boundary after a finite number of iterations and we have to stop.
  Thus, this operator still has some edges of the diagonal in its support.
  A visualization of these iterations can be seen in   \autoref{fig:visualization_planar_surface_linear_optic_mix}.
  Therefore, every codeword that corresponds to the usage of any non-trivial $X$-stabilizer generator on the diagonal allows a successful determination of the $\overline{XX}$ value.
  The diagonal consists of $2n-1$ edges and each $X$-stabilizer generator along the diagonal has two edges which lie on the diagonal.
  Therefore, our argument considers $2(n-1)$ $X$-stabilizer generators and every combination of them except using no stabilizer generator leads to success.
  As a consequence, the overall logical BM efficiency without photon loss is given by \[\frac{1}{2}\left(1+1-\frac{1}{2^{2(n-1)}}\right)=1-2\cdot 4^{-n}\,.\]\newline 
\end{proof}

\begin{figure}
\subfloat[]{%
  \includegraphics[width=0.25\textwidth]{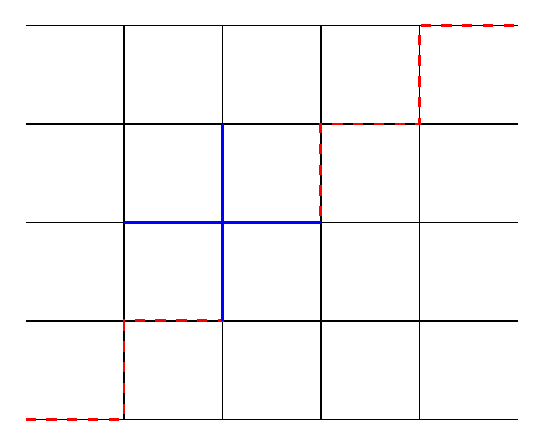}%
}
\subfloat[]{%
  \includegraphics[width=0.25\textwidth]{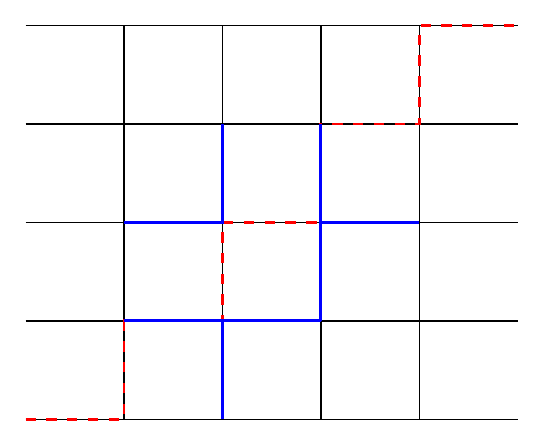}%
}

\subfloat[]{%
  \includegraphics[width=0.25\textwidth]{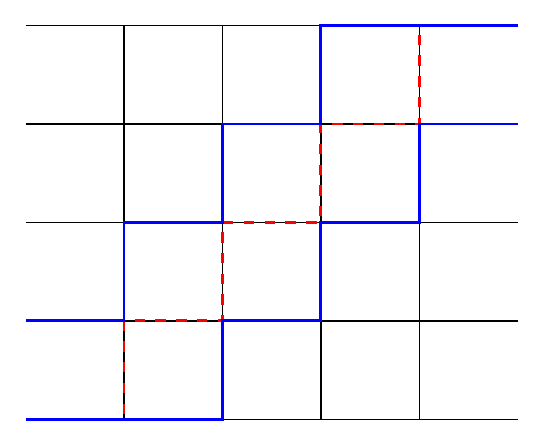}%
}

	\caption{Visualization of codewords that correspond to $X$-stabilizer generators within a surface(5,5) code. The dotted red line is a visual aid and shows the diagonal where we perform $ZZ$=1 measurements. In (a) we can see the support of an $X$-stabilizer generator (marked as blue lines) whose support contains edges of the diagonal. In (b) we can see the support of the operator after multiplying two adjacent $X$ stabilizer generators, resulting in a movement of the edges which lie on the diagonal and belong to the support of the operator. (c) shows that the procedure performed in (b) can only be applied a finite number of times due to the boundary of the plane.}
	\label{fig:visualization_planar_surface_linear_optic_mix}	
\end{figure}

We can also generalize this argument for surface codes with $n\neq m$.
Let us assume $n>m>1$. For $m>n>1$, everything works similarly when using the dual lattice and exchanging the guaranteed-information BMs. The formation of $ZZ$=1 measurements forms a zig-zag, which can be seen in \autoref{fig:toric_mix_34}. We need $m>1$ in order to be able to construct this zig-zag pattern.  The diagonal in the cases of $n=m$ is a special case of this zig-zag.
The actual argument in this calculation is the same as in the previous case with $n=m$ and we only need to count the number of $X$-stabilizer generators whose support includes the zig-zag. $2n-1$ edges lie on this zig-zag and thus we can use $2(n-1)$ $X$-stabilizer generators like in the previous argument. Since we can use a similar argument on the dual lattice if $m>n>$1 and using $XX$=1 measurements on the zig-zag, we obtain a logical BM efficiency in the absence of photon loss of  
\begin{equation*}
	1-2 \cdot 4^{-\max(n,m)} \hspace{0.2cm}.
\end{equation*}

\begin{figure}
	\centering
	\includegraphics[width=0.4\textwidth]{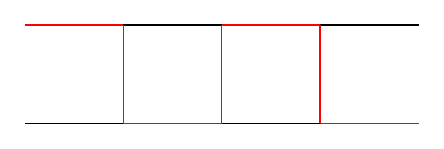} 
	\caption{One example of the formation of physical BMs with different guaranteed-information of a planar surface(n,m)-code, where $n>m>$1. Red edges correspond to guaranteed-$ZZ$=1 measurements and black edges correspond to guaranteed-$XX$=1 measurements, such that $\overline{ZZ}$-information is always given without loss. For the $\overline{XX}$ information only one of the red edges BMs need to give additional information. If $m>n>1$ one goes to the dual-lattice picture and gets the same lattice as in the Figure, but then red edges correspond to guaranteed-$XX$ BMs and black edges to guaranteed-$ZZ$ BMs. The diagonal path of red edges in the case $n=m$ is a special case of the zig-zag.}
	\label{fig:toric_mix_34}
\end{figure}

\section{Proof of \autoref{thm:qpc} }
\label{sec:qpcproof}
\begin{theorem*}
The no-loss BM efficiency of QPC($n,m$) with static linear optics is given by $1-2^{-\left(n+m-1\right)}$ when using the measurement formation as given in \autoref{fig:formations}(e).
\end{theorem*}
\begin{proof}
We are going to prove this result for the special case of $n=m=2$ and then we are generalizing it inductively.\newline
Notice that we always obtain $\overline{YY}=(i \overline{X}\overline{Z})(i \overline{X}\overline{Z})=\overline{XX}\overline{ZZ}$ information and we also obtain full logical information if the $YY$=1 measurement gives full information. 
Therefore, we will count the number of codeword combinations that do not allow for an identification of the logical Bell states.
A basis for $X$- and $Z$-codewords is obtained as usually by the support of stabilizer operators and the logical operators and by applying \autoref{lem:enum} to it. 
This time we do not use the obtained basis directly, but we use a linear combination of these basis elements to build a different basis for $X$- and $Z$-codewords, so that enumeration arguments will simplify. 
The $X$-codewords are given by span $(\underset{\circled{1}}{1_{(1,1)}+ 1_{(2,1)}},\underset{\circled{2}}{ 1_{(1,2)}+ 1_{(2,1)}},\underset{\circled{3}}{ 1_{(2,1)}+ 1_{(2,2)}})$ and the $Z$-codewords are given by  
 span $(\underset{\circled{$1^*$}}{1_{(1,1)}+ 1_{(1,2)}}, \underset{\circled{$2^*$}}{1_{(2,1)} +1_{(2,2)}})$. In this notion $1_{(j,k)}$ is a string $\in \mathbb{Z}_2^{n\cdot m}$ where the entry corresponding to the qubit pair in row j and column k is 1  and all other entries of this string are 0.
 We are able to discriminate all logical Bell states if we either get full information in (1,1) or (1,2) or ((2,1) and (2,2)). Codewords that involve either both of $\circled{1}$ and $\circled{$1^*$}$ or none of them give no full information on (1,1). Codewords that involve $\circled{2}$ give full information on (1,2) and codewords that involve $\circled{$2^*$}$ give full information on (2,1) and (2,2). In contrast $\circled{3}$ is independent of any information gain. Thus, there exist 4 out of $2^5$ codewords that do not allow for an identification of the logical Bell state.
 Therefore, this formation achieves a logical BM efficiency of $1-\frac{4}{2^5}=1-2^{-3}=1-2^{-(n+m-1)}$ for $n=m=2$.\newline
 Now suppose this also holds for arbitrary $n\geq$2 while $m$=2.
 Increasing $n\rightarrow n+1$ gives two new stabilizer generators.
  Namely, we get a new $X$-stabilizer generator with support on all qubits of the $n$\textsuperscript{th} and ($n+1$)\textsuperscript{th} row, which corresponds after linear combinations of previous $Z$-codewords to a new basis $Z$-codeword $1_{(n+1,1)}+1_{(n+1,2)}$. 
  Any codeword that involves this basis codeword gives $\overline{XX}$ information along the ($n+1$)\textsuperscript{th} row, which allows an identification of the logical Bell state. 
 We also get a new $Z$-stabilizer generator with support on the ($n+1$)\textsuperscript{th} row. Therefore, the resulting basis $X$- codeword only makes a difference on qubit pairs with $ZZ$=1 measurements and, as a consequence, it has no influence on the ability to perform a successful logical BM. Similarly, we have to extend the codeword corresponding to $\overline{Z}$ to the ($n+1$)\textsuperscript{th} row, which also has no influence on the ability to identify a Bell state with a $ZZ$ measurement.
 This means the number of all codeword combinations increases by a factor of 4 while the number of undecodable codeword combinations only increases by a factor of 2. 
 Therefore, we have shown that the logical BM efficiency is $1-2^{-(n+m-1)}$ for $m$=2, $n\geq$2, and now we fix $n$ and increase $m\rightarrow m+1$.

This gives us $n$ new $Z$-stabilizer generators corresponding to $n$ basis $X$-codewords. Using linear combinations we can obtain a basis $X$-codeword of the form $1_{(1,m+1)}+\dots+1_{(n,m+1)}$.
 Every codeword combination that involves the new basis $X$-codeword  gives $\overline{ZZ}$ information along the ($m+1$)\textsuperscript{th} column, allowing an identification of the logical Bell state. Furthermore, we obtain $n$-1  basis $X$-codewords of the form  $1_{(\text{j},m)}+1_{(\text{j},m+1})$ for  $j\in\{2,...,n\}$. However, these basis $X$-codewords only act on qubits where $ZZ$ measurements are performed and hence they have no influence on being able to identify Bell states. We also have to consider that $\overline{X}$ is extended to the ($m+1$)\textsuperscript{th} column and this  basis $Z$-codeword has an influence on the $ZZ$ measurements in the ($m+1$)\textsuperscript{th} column, but this has no effect on the ability to obtain $\overline{ZZ}$ information, because all $ZZ$ measurement within a row yield the same result.
 Therefore, the number of all codeword combinations increases by a factor of $2^n$ while the number of undecodable codeword combinations only increases by a factor of $2^{n-1}$.\newline As a result we can see that the logical BM efficiency without photon loss is given by $1-2^{-(n+m-1)}$ for  $n,m>1$
\end{proof}

\section{\label{a:nocloning}No-cloning}
\begin{theorem}\label{eq:nocloning}
	The logical transmission probability $\eta_{log}$ of a quantum code is not higher than $\frac{1}{2}$ for $\eta=\frac{1}{2}$ in the erasure channel.
\end{theorem}
\begin{proof}
	Recall that a general erasure channel
	\begin{equation*}
		\rho \rightarrow \eta \rho + (1-\eta)\ket{e}\bra{e}
	\end{equation*}
	can also be seen as a unitary evolution in an extended Hilbert space,
	\begin{equation}
		\begin{aligned}
		&U_{iso}:\mathcal{H}_A\rightarrow\mathcal{H}_B\otimes \mathcal{H}_E\\
		U_{iso}(\eta)=&\sqrt{\eta}\left(\ket{0}_B\bra{0}_A+\ket{1}_B\bra{1}_A\right)\otimes \ket{e}_E\\&+\sqrt{1-\eta}\ket{e}_B\otimes\left(\ket{0}_E\bra{0}_A+\ket{1}_E\bra{1}_A\right) \,,
		\end{aligned}
	\end{equation}
	where we may assign subsystem $A,B,E$ to Alice, Bob, Eve (environment), respectively.
	With $\eta=\frac{1}{2}$ the resulting state is symmetric with respect to both parties Bob and Eve and thus both can describe their part of the system with equal reduced density operators.
	Analogously one can show that Bob and Eve have equal density operators when considering a quantum code and applying $U_{iso}^{\otimes n}\left(\frac{1}{2}\right)$.
	Due to the no-cloning theorem it is impossible that Bob as well as Eve can successfully decode the quantum state encoded by Alice.
	Thus the probability of a successful decoding $\eta_{log}$	is not higher than $\frac{1}{2}$.
\end{proof}

\section{Construction of Charlie's probability distribution}
\label{app:probdistri}
Suppose we give Charlie the tuple $\{p_1,p_2,p_3,p_4\}$ containing our desired physical BM efficiencies for the four Bell states.
Calculating Charlie's probability distribution $p_A,\dots,p_F$ is equivalent to finding a solution to the linear system of equations, where $0\leq p_A,\dots,p_F$:
\begin{align}
	\begin{pmatrix}
		p_1\\
		p_2\\
		p_3\\
		p_4\\
		1
	\end{pmatrix}
	=\begin{pmatrix}
		0 & 0 & 0 & 1 & 1 & 1 \\
		1 & 1 & 0 & 1 & 0 & 0 \\
		1 & 0 & 1 & 0 & 1 & 0 \\
		0 & 1 & 1 & 0 & 0 & 1 \\
		1 & 1 & 1 & 1 & 1 & 1
	\end{pmatrix}
	\begin{pmatrix}
		p_A\\
		p_B\\
		p_C\\
		p_D\\
		p_E\\
		p_F
	\end{pmatrix} \,.
\end{align} 
Let us order the four Bell states in such a way that $p_1\leq p_2 \leq p_3 \leq p_4$ and let us now consider the special case of $p_1=0$.
Thus, we only have to consider three equations depending on the input probabilities plus one equation ensuring that the probabilities sum up to one.
By adding the three equations we obtain $p_C=\frac{1}{2}\left(p_4+p_3-p_2\right)\geq \frac{1}{2}p_4\geq 0$ where we used that we ordered the probabilities by magnitude.
By substituting $p_C$ we obtain $p_A=\frac{1}{2}\left(p_3+p_2-p_4\right)=\frac{2}{3}(1-p_4)\geq0$ using $p_1+p_2+p_3+p_4=2$.
Similarly, we obtain $p_B=\frac{1}{2}\left(p_4+p_2-p_3 \right)\geq0$,
$p_A+p_B+p_C=\frac{1}{2}(p_2+p_3+p_4)=1$.\\
Let us now exploit this result in order to consider the general case where $p_1>0$ is possible. 
We will show that it is, for example, possible to decrease $p_2$ and increase $p_1$ while $p_3,p_4$ remain constants. $p_2=p_A+p_B>0$. Let us assume that, for example, $p_2=p_A$, then we can shift probability $p$ from $p_2$ to $p_1$ by decreasing $p_A$ by $p$ and increasing $p_E$ by $p$ such that $p_3,p_4$ do not change. However, when $p_B>0$ it could be that we want to shift a $p>p_A$. Then we first shift with $p_A$ and then we perform the remaining shifting of $p_2$ on $p_B$ (increasing $p_F$). If we also need to shift probabilities from $p_3$ or $p_4$, one can perform similar steps.
Thus, we can summarize Charlie's protocol of finding his probability distribution:
\begin{itemize}
	\item Charlie receives a tuple $\{p_1,p_2,p_3,p_4\}$ describing the state-dependent desired efficiencies of the physical BMs.
	\item Charlie relabels the elements such that $p_1<\dots<p_4$.
	\item Charlie defines auxiliary probabilities $\widetilde{p}_1=0,\widetilde{p}_2\geq p_2,\widetilde{p}_3\geq p_3,\widetilde{p}_4\geq p_4$ and still $\widetilde{p}_2\leq \widetilde{p}_3\leq \widetilde{p}_4$.
	\item Charlie uses the solution for the $\widetilde{p}_1=0$-special case.
	\item Charlie shifts probabilities to $\widetilde{p_1}$ until he obtains the desired distribution  $\{p_1,p_2,p_3,p_4\}$.
\end{itemize} 

\section{Scripts}

All scripts that have been used for obtaining numerical results and plots in this paper can be found at \url{https://github.com/schmidtfrk/CSSlogBM} .

 \normalem
\bibliography{ref}

\end{document}